\documentclass[11pt]{article}
\usepackage[]{amsmath,amssymb,amsfonts,latexsym,amsthm,enumerate,fullpage}
\usepackage[utf8]{inputenc}
\usepackage{color}
\usepackage{amsmath,amssymb,amsthm}
\usepackage{graphicx}
\usepackage{tikz}
\usetikzlibrary{math,calc}

\newtheorem{thm}{Theorem}[section]
\newtheorem{prop}[thm]{Proposition}

\newtheorem{lemma}[thm]{Lemma}
\newtheorem{cor}[thm]{Corollary}

\newtheorem{definition}[thm]{Definition}

\newtheorem{rmk}[thm]{Remark}

\newtheorem{ex}[thm]{Example}
\newtheorem{obs}[thm]{Observation}

\newcommand{\R}{\mathbb{R}}
\newcommand{\X}{\mathbf X}
\newcommand{\Y}{\mathbf Y}
\newcommand{\Z}{\mathbb Z}
\newcommand{\F}{\mathbb F}

\newcommand{\Aut}{\text{Aut}}
\newcommand{\mus}{\mu^{\text{coll}}}
\newcommand{\tbeta}{\widetilde\beta}
\newcommand{\valpha}{\vec{\alpha}}
\newcommand{\nr}{\parallel}

\begin{document}
\title{New Cosystolic Expanders from Tensors Imply\\ Explicit Quantum LDPC Codes with $\Omega(\sqrt{n}\log^kn)$ Distance}

\author{
Tali Kaufman
\footnote{Department of Computer Science, Bar-Ilan University, Ramat-Gan, 5290002, Israel, email:kaufmant@mit.edu}
\and
Ran J. Tessler
\footnote{Department of Mathematics, Weizmann Institute of Science, POB 26, Rehovot 7610001, Israel, email:ran.tessler@weizmann.ac.il }
}

\maketitle
\begin{abstract}
In this work we introduce a new notion of expansion in higher dimensions that is stronger than the well studied cosystolic expansion notion, and is termed {\em Collective-cosystolic expansion}.

We show that tensoring two cosystolic expanders yields a new cosystolic expander, assuming one of the complexes in the product, is not only cosystolic expander, but rather a collective cosystolic expander.

We then show that the well known bounded degree cosystolic expanders, the Ramanujan complexes are, in fact, collective cosystolic expanders.
This enables us to construct new bounded degree cosystolic expanders, by tensoring of Ramanujan complexes.

Using our new constructed bounded degree cosystolic expanders we construct {\em explicit} quantum LDPC codes of distance $\sqrt{n} \log^k n$ for any $k$, improving a recent result of Evra et. al. \cite{EKZ}, and setting a new record for distance of explicit quantum LDPC codes.

The work of \cite{EKZ} took advantage of the high dimensional expansion notion known as cosystolic expansion, that occurs in Ramanujan complexes. Our improvement is achieved by considering tensor product of Ramanujan complexes, and using their newly derived property, the collective cosystolic expansion.
%
\end{abstract}

\section{Introduction}
A cosystolic expansion is a topological measure of expansion in high dimensions which is not equivalent to the spectral high dimensional measure of expansion captured by the well studied notion of local-spectral expansion.
This topological expansion is hard to obtain.  There is essentially only one known family of bounded degree complexes that possesses this topological notion of expansion. This family is derived from the so called \emph{Ramanujan complexes} arising from number theory \cite{LSV1,LSV2}. In dimension two another family of cosystolic expanders was recently found in \cite{KO2}.
Part of the interest in this topological notion of high dimensional expansion, is related to the strong connection between cosystolic expanders and quantum LDPC codes with good distance.

In this work we develop a theory showing conditions under which the tensor product of two cosystolic expanders yields a new cosystolic expander.
Interestingly, to get cosystolic expansion for the product, we need to introduce a stronger notion of expansion than cosystolic expansion, that we term {\em collective cosystolic expansion}.
In a nutshell, for the product of cosystolic expanders to yield a new cosystolic expander, one of the complexes in the produce needs to possess this stronger expansion property, namely to be a collective cosystolic expander.

We further show that the well known Ramanujan complexes are not only cosystolic expanders, as was shown by \cite{KKL,EK}, but rather they are also collective cosystolic expanders. Hence by taking tensor products of them, we obtain many new families of bounded degree cosystolic expanders, in all dimensions.

We then use the new cosystolic expanders that we have constructed to construct {\em explicit} LDPC quantum code of distance $\sqrt{n} \log^k n$ for any $k,$ thus improving over the state of the art explicit quantum LDPC codes obtained recently by \cite{EKZ}.
The current distance record for not necessarily explicit LDPC quantum codes is $n^{\frac{3}{5}},$ achieved in the recent exciting work \cite{HHO}, which constructs a {\em random} LDPC code from some kind of a random combinatorial $S^1$-bundle.

The \cite{EKZ} paradigm obtained explicit quantum codes from cosystolic expanders (and more precisely Ramanujan complexes); It further implied that if it were possible to derive the quantum codes from homologies and cohomologies of large dimensions, their distance could have also been $\sqrt{n}\text{polylog}(n),$ with powers in the polylog that are as large as one wishes. However, it is not known whether the cohomologies are non zero in dimensions greater than $2.$ This was the bottleneck in improving the distance of the code constructed by \cite{EKZ} from $\sqrt{n}\log{n}$ to $\sqrt{n}\text{polylog}(n)$. \cite{EKZ} leaves it as an open question, whether this bottleneck can be addressed. Our idea of obtaining the improved code, in a nutshell, is the following: by tensoring we can ensure non vanishing cohomologies in dimensions as high as we wish, while maintaining the cosystolic expansion. This results in the improved quantum LDPC codes that we get.

\subsection{High dimensional expanders, cosystolic expanders and collective cosystolic expanders}
In the following we discuss high dimensional expanders with emphasis on the topological notion of expansion that is known as cosystolic expansion. Then we introduce the stronger notion of collective cosystolic expansion.
This new notion is going to play a key role in the tensoring machinery that we are developing in this work.

\paragraph{High dimensional expanders.} High dimensional expansion is an emerging theory which gains considerable attention recently as it led to solutions of several important open questions in theoretical computer science and mathematics. To list two notable examples: In mathematics, the study of high dimensional expanders led to a resolution \cite{KKL,EK} of an open question raised by Gromov \cite{Gromov,FGLNP}, as to whether topological overlapping property is possible in bounded degree high dimensional complexes. In computer science, this theory led to a solution of a famous conjecture of Mihail and Vazirani about the existence of an efficient algorithm to count the bases of a matroid for any matroid. This was recently achieved by \cite{AKSV}, building on high dimensional random walks that were studied in \cite{KO}.

The high dimensional expansion theory generalizes the well studied and influential theory of graph expansion to higher dimensions. It turns out that high dimensional expansion is a much more rigid phenomenon than expansion in graphs, and this implies some of the important aspects of this theory. There are several non-equivalent definitions to what is a high dimensional expander. In this work we refer to the notion of cosystolic expansion that was defined in \cite{EK}. Cosystolic expansion is generalisation of the topological notion of expansion in graphs, known as the Cheeger constant of a graph, to higher dimensions. There is a very closely related definition of {\em coboundary expansion} that was defined by Gromov and by Linial and Meshulam \cite{Gromov,LM2}. However, in this work we refer to the exact definition of cosystolic-expansion that is relevant for obtaining quantum CSS codes.

\begin{definition}[Cosystolic expansion \cite{EK}]
Given $\eta,\mu > 0,$ a $d$-dimension complex $X$ with a weight function $\nr-\nr$ (e.g. the normalized Hamming norm) on its cochains is a $(\mu,\eta)$-cosystolic expander if for every $i<d$
\begin{itemize}
\item The $i$-cofilling constant at most $\mu,$ i.e., for every non-zero $i$-cochain $\alpha$, whose boundary is the $i+1$ cochain $\delta(\alpha)$, there exists an $i$-cocycle $\alpha'$ such that $\nr \alpha-\alpha'\nr  \leq \mu\nr\delta(\alpha)\nr.$
\item The weight of the $i$-cosystole is at least $\eta$.
\end{itemize}
\end{definition}

We now move to introduce our new notion of topological expansion that is stronger than cosystolic expansion and is termed collective cosystolic expansion:

\begin{definition}[Collective cosystolic expansion]
Given $\eta, \mus > 0$, a $d$-dimension complex $X$ with a weight function $\nr-\nr$ on its cochains is a $(\mus,\eta)$-collective cosystolic expander if for every $i<d$
\begin{itemize}
\item The $i$-collective cofilling constant is at most $\mu$.

The $i$-collective cofilling constant of a complex $\X$ is defined as
\[\mus_i(\X)=\max_{\beta_1,\ldots,\beta_m\in B^{i+1}(\X)}\min_{\substack{\alpha_1,\ldots,\alpha_m\in C^i(\X):\\
\delta(\alpha_i)=\beta_i}}
\left\{\frac{\parallel\bigcup_{a\in[m]}\alpha_a\parallel}
{\parallel\bigcup_{a\in[m]}\beta_a\parallel}\right\}. \]

Note that restricting to the case that $m=1$ gives rise to the usual cosystolic expansion.
\item The weight of the $i$-cosystole is at least $\eta$.
\end{itemize}
\end{definition}
\begin{rmk}
In both definitions, in case there are no cosystoles, the complexes are termed $\mu-$coboundary expanders and $\mus-$collective coboundary expanders respectively.
As usual in the context of expanders, when a family of complexes is said to be a (collective) cosystolic expander, it is with respect to uniform $(\mu,\eta)$ (or $(\mus,\eta)$), and similarly for (collective) coboundary expanders.
\end{rmk}


\subsection{Tensoring cosystolic expanders via collective cosystolic expansion property}
In the following we show how to obtain new cosystolic expanders from tensoring two cosystolic expanders where one of the complexes in the product possesses the stronger collective cosystolic expansion property. In fact we have a hierarchy of results, ordered by the strength of the assumptions on the expansions of the complexes in the tensor product. The first level in the hierarchy deals with the product of a cosystolic expander and a complex with linear cosystoles, the second level deals with the product of a cosystolic expander and a collective cosystolic expander, and the third level deals with the product of two collective cosystolic expanders.

We state an approximate version of the theorem now, referring the reader to Theorem \ref{thm:linear_cosys} for the precise, and more general, statement.

\begin{thm}[Tensoring Cosystolic Expanders Theorem]\label{thm:1}
Let $\X,\Y$ be two complexes of bounded degree and dimension greater than $d.$
\begin{itemize}
\item Suppose in addition that $\X$ is a cosystolic expander and $\Y$ has cosystoles of linear size. Then the tensor product complex $\X\otimes\Y$ is of bounded degree and has linear cosystoles in any dimension smaller or equal to $d.$
\item Suppose in addition that $\X$ is a cosystolic expander and $\Y$ is a collective cosystolic expander. Then the tensor product complex $\X\otimes\Y$ is a bounded degree cosystolic expander in any dimension smaller or equal to $d.$
\item Suppose in addition that both $\X$ and $\Y$ are collective cosystolic expanders. Then the tensor product complex $\X\otimes\Y$ is a bounded degree collective cosystolic expander in any dimension smaller or equal to $d.$
\end{itemize}
\end{thm}
This theorem shows that importance of collective cosystolic expansion, and that this is the correct notion when one wants to consider products of high dimensional expanders.

\subsection{Ramanujan complexes are bounded degree collective cosystolic expanders}
In the following we briefly discuss Ramanujan complexes, which until this work were essentially the only family of bounded degree high dimensional complexes that were proven to be cosystolic expanders. In the recent work \cite{KO}, another family of bounded degree complexes was studied, and, in dimension two this family was shown to be a family of cosystolic expanders. 

\paragraph{Ramanujan Complexes.} Bounded degree graphs that are expanding are abundant, e.g., a random d-regular graph is such. However, high dimensional and bounded degree cosystolic expanders seem to be rare. A well studied family of such complexes, known as the family of Ramanujan complexes, was constructed in \cite{LSV1,LSV2}. In the work of \cite{KKL,EK} it was shown that the Ramanujan complexes are explicit bounded degree cosystolic expanders.

%

One of the main findings of our work here is that the Ramanujan complexes are not only cosystolic expanders but rather they are collective cosystolic expanders. For showing that, we also show that their building blocks, the \emph{spherical buildings} are also collective cosystolic expanders with vanishing cohomology, i.e. collective coboundary expanders. The main technical novelty in our proof that the Ramanujan complexes are collective cosystolic expanders stems in using the new notion of \emph{locally minimal collection of cochains} which is a generalization of the local minimality notion introduced in \cite{KKL,EK}.

The basic definitions required for the proof that the Ramanujan complexes are collective cosystolic expanders, as well as the proof, will be given in Section \ref{sec:strong_coff_LSV}. The proof of the analogous property for spherical buildings appears in Appendix \ref{app:spherical}.

\begin{thm}[Ramanujan Complexes are collective cosystolic expanders Theorem]\label{thm:strong_cofilling_Ramanujan}
For any $d\in\mathbb{N},$ any large enough $q,$ there exists a constant $\mus(d,q)$ such that and any $0\leq k<d,$ all rank $d+1$ Ramanujan complexes whose codimension-$1$ cells touch $1+q$ top cells are collective cosystolic expanders with collective cofilling constant at most $\mus(d,q).$
\end{thm}

\begin{thm}[Spherical Buildings are collective coboundary expanders Theorem] \label{thm:strong_cofilling_spherical}
For any $0\leq k\leq n-1$ there exists a constant $\mus(k,n)$ such that any rank-$n+1$ spherical building is a collective coboundary expander whose collective cofilling constant at most $\mus(k,n).$
\end{thm}

\paragraph{On cosystolic expansion vs. collective cosystolic expansion}
The collective cosystolic expansion is obviously stronger than the usual cosystolic expansion. We will show, however, that both spherical complexes and Ramanujan complexes have this property, and with the same proven constants.
This raises the question
\\\textbf{Open question.} Does any cosystolic expander is also a collective cosystolic expander with same constants?

\subsection{New bounded degree cosystolic expanders from tensors of Ramanujan complexes}
In the following we describe new families of bounded degree complexes which are cosystolic expanders, and have non vanishing cohomology groups in dimensions greater than $2.$ Such families were not known to exist prior to this work, and are obtained from tensors of Ramanujan complexes. 
The non vanishing of the high cohomologies will be crucial to obtaining the improved quantum codes we present.

The new bounded degree cosystolic expanders that we construct here are obtained by taking a tensor product of a $d$ dimensional Ramanujan complex with itself $\ell < d$ times. Since we have shown that Ramanujan complexes are collective cosystolic expanders we can apply our theorem about tensoring collective cosystolic expanders to the Ramanujan complexes, to extract from them new bounded degree cosystolic expanders with cosystoles of linear size.\footnote{We will work with the normalized Hamming norm. With this weight function, the second item in the definition of cosystolic expander implies that the cosystole is linear.}

The new bounded degree complexes that are obtained by the tensoring operation have additional two properties that we wish to highlight. These additional two properties will be useful later on in the quantum code construction that we introduce.

The first property is the non vanishing of cohomologies up to level $k \leq \ell$. The second property is a systolic distance that grows with $\ell$.

\paragraph{Non vanishing of cohomologies up to level $k \leq \ell$.}  In \cite[Propositions 3.5,3.6]{KKL} it was shown that for any $d\geq 1$ there exist infinitely many Ramanujan complexes of dimension $d$ for which the $1$-st cohomology does not vanish, and for any $d\geq 2$ there exist infinitely many Ramanujan complexes of dimension $d$ for which the $i$-th cohomology does not vanish for $i=1,2.$

By the properties of the Ramanujan complexes mentioned above, it has non-zero cohomologies at dimension $1$. We use this to deduce that in our new co-systolic expander, that was obtained by taking the $k-$skeleton of the $k$-tensor-power of a Ramanujan complex, the cohomologies do not vanish up to level $k$.

\paragraph{The systolic distance of the tensored complex.} The injectivity radius of the Ramanujan complex is the maximal $r$ such that for any vertex in the complex, the ball of radius $r$ around it, is a contractible subcomplex. Since the Ramanujan complexes are quotients of Bruhat-Tits buildings, this radius is also the minimal distance between vertices identifies by the quotient, minus $1.$
A consequence of \cite[Proposition 3.3]{LM} (see also \cite[Theorem 5.10]{EKZ}) is that Ramanujan complexes have injectivity radius which is logarithmic in the size of the complex.
Using it \cite[Theorem 5.11]{EKZ} showed that the $i$-systole of a Ramanujan complex of dimension $d > i$ is of length at least $\log^{i} n,$ assuming the systole exists.

We show that the same estimate holds for the systoles of the $k-$tensor power of a Ramanujan complex with logarithmic injectivity radius, only that in this case the existence of systoles is guaranteed, as explained before.

The following theorem summarizes the main properties of the tensor powers of Ramanujan complexes.

\begin{thm}[New bounded degree cosystolic expanders by tensoring Ramanujan complexes]\label{thm:2}
Let $\X$ be a Ramanujan complex with $n$ vertices, dimension $d$, locality bounded by $Q$ and injectivity radius at least $c\log{n}$ for some constants $c,Q>0,$ and non vanishing first cohomology. Let $\Y$ be the complex obtained by taking the tensor product of $\X$ with itself $l<d$ times. Then $\Y$ is a bounded degree cosystolic expander with the following properties:
\begin{itemize}
\item {\bf Bounded degree property.} $\Y$ has $n^{l}$ vertices, locality at most $lQ$ and hence all nonempty $Y_i$ are of size $\Theta(n^{l}),$ where the coefficient in $\Theta$ depends only on $l,Q.$
\item {\bf Collective-cosystolic expansion property.} The $k-$cofilling constant of $\Y,$ for $k\leq l$ is at most $\mu,$ for a constant $\mu$ which depends only on $Q,l,$ the cosystolic bounds of $\X$ and the cofilling constants of $\X.$
\item {\bf Non vanishing cohomologies up to dimension $\ell$.} $H_k(\Y),H^{k}(\Y)\neq 0,~k\leq l.$
\item {\bf Linear cosystole size.} The Hamming norm of the $k$th cosystole of $\Y,$ for $k\leq l$ is $\Omega(|Y_k|),$ where the coefficient inside $\Omega$ depends only on $Q,l,$ the cosystolic bounds and cofilling constants of $\X.$
\item {\bf Systoles grow with the tensor power.} $\text{Sys}_k(\Y)=\Omega(log^{k}(|Y_k|))$ for $k\leq l,$ where the constant inside $\Omega$ depends only on $k,l,d,c.$
\end{itemize}
\end{thm}
Again we prove a more general version below, see Theorem \ref{thm:complexes_with_linear_cosys_and_log_power_sys}

\subsection{New Quantum LDPC codes from tensoring Ramanujan complexes}

In the following we construct explicit quantum LDPC codes (CSS codes) 
using the new bounded degree cosystolic expanders that we have constructed by tensoring Ramanujan complexes. Our codes achieve distance of $\sqrt{n} \log^k n$ for any $k$. This sets a new record on the distance for {\em explicit} quantum LDPC codes.
The very recent breakthrough \cite{HHO} provides a random quantum LDPC code of distance $n^{3/5},$ which is the current record for quantum LDPC codes, but that construction is non explicit.

%

The work of \cite{EKZ} introduced the usefulness of high dimensional expanders and, in particular, of the Ramanujan complexes to the construction of quantum LDPC codes. Our improvement is essentially derived by turning to consider tensor products of high dimensional expanders, and in particular tensor product of Ramanujan complexes.

Below we discuss what a CSS quantum codes are and the \cite{EKZ} paradigm for constructing LDPC quantum code from co-systolic expanders. We then discuss a bottleneck in their paradigm and how we overcome it to get our improved codes.

\paragraph{On CSS codes and $i$-(homology, cohomology)-pair.} A quantum CSS code can be naturally obtained from pairs of dual spaces, the $i$-th homology and cohomology spaces, associated with a $d$ dimensional complex (for any $0<i<d$). The dimension of the derived CSS code is the dimension of the $i$-th homology, which equals the dimension of the $i$-th cohomology. The minimal length element in the $i$-cohomology is called the $i$-cosystole and the minimal length element in the $i$-homology is called the $i$-systole. The distance of the quantum CSS code associated to a pair of $i$-homology and  $i$-cohomology is the minimum of the length of the $i$-systole and $i$-cosystole.  For detailed definition of these topological notions: homology, cohomology, systoles and cosystoles, we refer the reader to Section \ref{sec:hom_cohom_etc}. For exact definition of quantum CSS codes we refer the reader to Section \ref{sec:css}.

\paragraph{The \cite{EKZ} paradigm, its bottleneck and our insight.} The discussion above, makes it clear that for constructing improved quantum codes one has to seek complexes with $i$-systoles and $i$-cosystole as large as possible. The first idea of \cite{EKZ} is that high dimensional expanders are going to provide complexes with linear length cosystoles and mildly large systoles. Their second idea, based on Hastings's work \cite{Has17}, is that it is possible to balance the distances between the $i$-systoles and $i$-cosystole of a certain complex to get a quantum LDPC code, whose distance is the geometric average of the two, and its rate is large if the $i$-th cohomology does not vanish. The relevant balancing theorem is the following

\begin{thm}[Balancing Theorem of \cite{EKZ} based on \cite{Has17}]\label{thm:balancing}
A complex of size $\Theta(m)$ with non zero $i$-th cohomology, with $i-$systole of size $\text{Sys}_i,$ and $i-$cosystole of size $\text{CoSys}^{i}\gg\text{Sys}_i,$ gives rise to a quantum LDPC code of length $n=\Theta(m\frac{\text{CoSys}^{i}}{\text{Sys}_i}),$ distance $\text{CoSys}^{i},$ and dimension
$\Omega(\frac{\text{CoSys}^{i}}{\text{Sys}_i}).$

In the special case that the $i-$cosystole is also linear in $m,$ the resulting code has length \\$n=\Theta(m^{2}/\text{Sys}_i),$ distance
\[\Theta(m)=\Theta(\sqrt{n\cdot \text{Sys}_i}),\] and dimension \[\Omega(\frac{m}{\text{Sys}_i})=\Omega\left(\sqrt{\frac{n}{\text{Sys}_i}}\right).\]
\end{thm}

\cite{EKZ} observed that as one turns to $i$-th homology and cohomology pair of larger and larger $i$, in a Ramanujan complex of dimension $d>i$, the distance of the derived quantum code improves, \emph{as long as its dimension does not vanish}. However, by increasing $i$ one loses control on the rate, and thus can not ensure that it is non zero. This is the bottleneck that \cite{EKZ} faced, and this what prevented them from improving the quality of the constructed codes. They leave it as an open question whether this bottleneck can be addressed. So \cite{EKZ} are facing the so called "High Cohomologies Challenge", namely if they knew that Ramanujan complexes have non vanishing cohomologies in large dimensions they could have used this knowledge to derive improved explicit quantum LDPC codes.

Our insight is that when we turn to tensor products of Ramanujan complexes, we can enjoy the effect of improving the code distance by turning to the $i$-th homology and cohomology pair of larger and larger $i$, while being able to ensure non vanishing rate. This is how we get our improved codes.
Namely our idea is the following: We consider a $k-$tensor power of a Ramanujan complex of degree $d>k.$
We argue that by the K\"unneth formula in the resulting complex the $k$-th cohomology will not vanish if the $1$-st cohomology of the original complex did not vanish. Thus, we get a new complex whose $k$-th cohomology is non zero. We use the collective cosystolic expansion property of Ramanujan complexes to deduce that the tensored complex is a cosystolic expander up do dimension $k,$ and hence its $k$-cosystole's length is linear in $n,$ where $n$ now stands for the number of vertices in the tensor power complex, and that its $k-$systole's length is $\Omega(\log^{k} n)$. Then, after applying the balancing procedure of Theorem \ref{thm:balancing}, we obtain the promised quantum codes.

\begin{cor}[Improved Quantum codes Theorem]\label{cor:good_codes}
For every $k$ there exist a LDPC quantum CSS code of distance $\Omega(\sqrt{n}\log^{k/2}(n))$ and dimension $\frac{\sqrt{n}}{\log n^{k/2}},$ where $n$ is the size of the code.\end{cor}

\begin{rmk}
Tensor products of quantum codes, as well as other closely related notions of products, were considered in literature (see, for example, \cite{AC,FH,BH,TZ}). Neither the systolic distance, nor the cosystolic distance are expected to be well behaved with respect to these operations in general. As explained above we exploit the properties of cosystolic expanders in order to analyze the cosystoles in the product code. For the analysis of systoles we need the special properties of the Ramanujan complex.
\end{rmk}

\subsection{Plan of the paper}
This paper is structured as follows. In Section \ref{sec:backg} we provide basic background in homological algebra, CSS codes and Ramanujan complexes. We also define systoles, cosystoles, cosystolic expansion and the new property of collective cosystolic expansion. In Section \ref{sec:prod} we prove Theorem \ref{thm:complexes_with_linear_cosys_and_log_power_sys}, which generalizes Theorem \ref{thm:1}, and discusses the cosystolic behavior of tensor products of complexes. In Section \ref{sec:strong_coff_LSV} we give a criterion for simplicial complexes which guarantees collective cosystolic expansion, and prove that Ramanujan complexes satisfy it, hence proving Theorem \ref{thm:strong_cofilling_Ramanujan}. Finally in Section \ref{sec:tensoramanujan} we prove Theorem \ref{thm:2} and construct the promised quantum LDPC codes.
Appendix \ref{app:spherical} briefly discusses spherical complexes, and proves a generalized version of Theorem \ref{thm:strong_cofilling_spherical}.
\subsection{Acknowledgments}
R.~T. would like to thank Or Hershkovits for helpful discussions.

R.~T. (incumbent of the Lillian and George Lyttle Career Development Chair) was supported by the ISF grant No. 335/19 and by a research grant from the Center for New Scientists of Weizmann Institute.

T.~K. was supported by ERC.

\section{Background and definitions}\label{sec:backg}
\subsection{Complexes, homology and cohomology}\label{sec:hom_cohom_etc}
A (bounded) {\em chain complex} over a coefficient ring $R$ is a collection of $R-$modules $\mathcal{C}_\bullet=(\mathcal{C}_m,\mathcal{C}_{m+1},\ldots,\mathcal{C}_n)$ and \emph{boundary} maps
$\partial_k:\mathcal{C}_k\to\mathcal{C}_{k-1}$ ($\mathcal{C}_{m-1}$ is taken to be $0$), which satisfy
\[\partial_k\circ\partial_{k+1}=0,~~\forall k.\]

We will restrict ourselves to $R=\F_2,$ the field with two elements, so that each $\mathcal{C}_k$ will be a $\F_2-$vector space, and $\partial_k$ a linear map between $\F_2-$vector spaces.
We also consider the dual \emph{coboundary} maps,
\[\delta^{k}=\partial^{*}_{k+1}:\mathcal{C}^{*}_k\to\mathcal{C}^{*}_{k+1}.\]
We call the elements of $\mathcal{C}_k$ \emph{$k-$chains}, and those of $\mathcal{C}^{*}_k$ \emph{$k-$cochains}.

Write \[Z_k(\mathcal{C})=\text{ker}(\partial_k),~~B_k(\mathcal{C})=\text{im}(\partial_{k+1}),
~~Z^{k}(\mathcal{C})=\text{ker}(\delta^{k}),~~,B^{k}(\mathcal{C})=\text{im}(\delta^{k-1}).\]
Elements of these sets are called $k-$cycles,~$k-$boundaries,~$k-$cocycles and $k-$coboundaries respectively.

The equations \[\partial_k\circ\partial_{k+1}=0, ~~\delta^{k+1}\circ\delta^{k}=0\] (the second is obtained by dualizing the first) imply that
\[B_k(\mathcal{C})\subseteq Z_k(\mathcal{C}),~~B^{k}(\mathcal{C})\subseteq Z^{k}(\mathcal{C}).\]
We define the $k-$homology and $k-$cohomology group of $\mathcal{C}$ (with $\F_2-$coefficients) as
\[H_k(\mathcal{C})=Z_k(\mathcal{C})/B_k(\mathcal{C}),~~H^{k}(\mathcal{C})=Z^{k}(\mathcal{C})/B^{k}(\mathcal{C})\]
respectively.
The $k-$th Betty number is defined as $h_k(\mathcal{C})=\dim_{\F_2} H_k(\mathcal{C}),$ which is readily seen to be equal $\dim_{\F_2} H^{k}(\mathcal{C})$ as well.

A \emph{based chain complex} is a chain complex where each $\mathcal{C}_k$ comes equipped with a distinguished basis. This is equivalent to writing $\mathcal{C}_k=\F_2^{X_k},$ for a set $X_k$ to which we will refer as the set of \emph{$k-$cells}. We shall usually denote the based chain complex by the sets $\X=(X_m,\ldots,X_n),$ and put $C_k(\X)=\F_2^{X_k}.$
We will sometimes lighten notations by writing $\sigma\in\X$, instead of $\sigma\in\bigcup X_i,$ for a cell $\sigma.$
Working with bases allows identifying $\mathcal{C}_k$ with its dual, hence considering both the $k-$chains and $k-$cochains as living in the same space. In addition it provides a natural definition of a \emph{support} and \emph{norm} of a chain or cochain $\alpha.$
The support of $\alpha,~\text{supp}(\alpha)$ is the set of cells $X_k$ which appear with coefficient $1$ when $\alpha$ is written according to the $X_k$ basis. We identify chain or cochains with their supports, and use it to define set theoretic operations such the union or intersection of chains or cochains, just by the corresponding operators on the support. We also sometimes write $\sigma\in\alpha$ to denote $\sigma\in\text{supp}(\alpha);$ in this case we say that the cell $\sigma$ belongs to $\alpha.$ The norm of $\alpha,~|\alpha|$ is the cardinality of its support, or equivalently the Hamming weight of $\alpha$ written in the $X_k$ basis.
For a cochain or a chain $\alpha$ denote by $\alpha_k$ its $C^k(\X)-$part.

We say that a cell $\sigma\in X_k$ is \emph{contained} in a cell $X_{k+l}$ if
\begin{itemize}
\item $l=1$ and $\sigma\in\text{supp}(\partial\tau).$
\item $l>1$ and $\sigma\in \text{supp}(\partial\sigma')$ for some $\sigma'$ which is contained in $\tau.$
\end{itemize}
A complex is \emph{pure} if each cell is contained in some top dimensional cell.

Many examples of based complexes come from topology. Two examples are simplicial and polyhedral complexes, another example we shall not need is CW-complexes.
A simplicial complex $\X=(X_0,\ldots,X_d)$ is a based chain complex where the elements of $X_k$ correspond to different subsets $S$ of $X_0,$ of size $k+1.$ The boundary map takes $S$ to the sum of its $k+1$ subsets of size $k.$ A simplicial complex always has lower locality $k+1$ in dimension $k.$ For simplicial complexes, a cell $\sigma$ is contained in a cell $\tau$ precisely if $\sigma$ is contained in $\tau$ as a set.
A generalization of a simplicial complex which we call polyhedral complex is\footnote{this is not the standard definition of a polyhedral complex, but it will suffice for our needs} a based chain complex where any $c\in X_k$ corresponds to subset of the \emph{vertex set} $X_0.$ And if $c\in X_k$ corresponds to a set $S,$ then $\partial_kc$ is a sum over several subsets of $S$ which correspond to elements of $X_{k-1}.$ We also require that all cells in $X_1$ correspond to pairs of elements of $X_0,$ which we call \emph{edges}, while $k-$cells for $k>1$ correspond to sets of vertices with more than $2$ elements.
In these two examples we identify the cells with the sets of vertices to which they correspond.
Polyhedral complexes arise naturally, as cubical complexes, and, more importantly for this work, as tensor products of simplicial complexes.
The dimension of a simplicial or polyhedral complex is the largest $d$ for which $X_d\neq \emptyset.$

We say that the based complex $\X$ has \emph{upper locality} $p_k$ in dimension $k$ if
\[sup_{v\in C^k(\X)}\frac{|\delta  v|}{|v|}\leq p_k,\]
this is equivalent to $sup_{\sigma\in X_k}|\delta \sigma|\leq p_k.$
We similarly say that $\X$ has \emph{lower locality} $q_k$ in dimension $k$ if
\[sup_{v\in C_k(\X)}\frac{|\partial  v|}{|v|}\leq q_k,\]
and again it is equivalent to requiring $sup_{\sigma\in X_k}|\partial \sigma|\leq q_k.$

We say that the complex is $Q-$\emph{local} if it has upper and lower localities $Q$ in each dimension.

A weighted complex is a based complex $\X$ together with a weight function\[\nr-\nr:C^\bullet(\X)\to\R_{\geq 0},\] which vanishes only on the $0-$cochain and is additive in the sense that
\[\nr\alpha\nr=\sum_{\sigma\in\text{supp}(\alpha)}\nr\sigma\nr.\]

The Hamming norm is an example of a weight function, another example is the normalized Hamming norm, defined by
\[\nr\alpha\nr=\frac{|\text{supp}(\alpha)|}{|X_k|},~\alpha\in C^k(\X).\]
A third example that will be relevant to Section \ref{sec:strong_coff_LSV} and in Appendix \ref{app:spherical}, is the weight function on $d-$dimensional complexes, given by
\begin{equation}\label{eq:lob_norm}\nr\alpha\nr=\frac{\sum_{\sigma\in\text{supp}(\alpha)}|\{\tau\in X_d\big|\tau ~\text{contains }\sigma\}|}
{\binom{d+1}{k+1}X_d},~~\alpha\in C^{k}(\X).\end{equation}

\subsection{Cosystoles, Systoles, Cosystolic expansion and Collective cosystolic expansion}
In the following we define systoles and cosystoles and then we will be ready the present the main definitions of this work: the well studied cosystolic expansion notion as well as a the stronger notion of collective-cosystolic expansion.

\bigskip

We define the \emph{$k-$systole} and \emph{$k-$cosystole} of a weighted complex as
\[\text{Sys}_k(\X)=\min_{\alpha\in Z_k\setminus B_k}\nr\alpha\nr,~~\text{CoSys}^{k}(\X)=\min_{\alpha\in Z^{k}\setminus B^{k}}\nr\alpha\nr.\]
We also define the $k-$th \emph{cosystolic bound} by
\[\eta_k(\X)=\text{CoSys}^{k}(\X)/\nr X_k\nr,\]
where we consider $X_k$ as the cochain whose support is $X_k.$
We make the convention that $\eta_k(\X)=1$ if $h_k(\X)=0.$
We define the $k-$th \emph{cofilling constant} by
\[\mu_k(\X)=\sup_{\alpha\in B^{k+1}(\X)\setminus\{0\} }\inf_{\substack{\beta\in C^k(\X):\\\partial \beta=\alpha}}\frac{\nr\beta\nr}{\nr\alpha\nr}.\]

A simple generalization of the cofilling constant that will play a key role in this work is the following notion.
\begin{definition}
The $k-$th \emph{collective cofilling constant} of the weighted complex $\X$ is defined by
\[\mus_k(\X)=\sup_{\alpha_1,\ldots,\alpha_m\in B^{k+1}(\X)\setminus\{0\} }\inf_{\substack{\beta_1,\ldots,\beta_m\in C^k(\X):\\\forall a~\partial \beta_a=\alpha_a}}\frac{\nr\bigcup_a\beta_a\nr}{\nr\bigcup_a\alpha_a\nr}.\]
\end{definition}
Clearly $\mu_k\leq\mus_k.$

We say that $\X$ \emph{has the (collective) cofilling property} in dimension $k$ with constant $c_k$ if its (collective) cofilling constant is bounded by $c_k.$ We say that $\X$ has the (collective) cofilling property with constants $\{c_k\}_k$ if $c_k$ upper bounds its (collective) cofilling constant in dimension $k$. In the special case that all (collective) cofilling constants are bounded by $c$ we say that $\X$ has the (collective) cofilling property with constant $c.$

We can define also systolic bounds and filling constants, but they will not be needed in this work.

\begin{rmk}\label{rmk:different weights}
If the ratio between weight functions is bounded from above and from below by two positive constants $R>r$, then the ratios of the cosystolic bound, cofilling and collective cofilling constants with respect to the two weight functions are bounded from above and below by $R/r$ and $r/R$ respectively.

For pure $d$ dimensional complexes of locality $q,$ the Hamming and normalized Hamming norms, and the weight function \eqref{eq:lob_norm}, give rise to cosystolic bounds, cofilling and collective cofilling constant whose ratios are bounded from below and from above by explicit positive constants which depend on $q,d.$ Thus, if we consider $q,d$ as constants themselves, we can estimate the cosystolic bounds, cofilling and collective cofilling constant for one weight function using the other, up to a multiplicative constant.

In fact, the cosystolic bounds calculated with respect to the Hamming and normalized Hamming weights are the same, the filling and cofilling constants are proportional with proportionality constant in dimension $k$ being $|X_k|/|X_{k+1}|.$

The only complexes which we will examine in this work, for which $q$ cannot be considered as a constant, are the spherical buildings discussed in the appendix. There it will be important to work with the weight function \eqref{eq:lob_norm}.
\end{rmk}

\subsection{Codes, CSS, etc}\label{sec:css}
\newcommand{\bfX}{\mathbf X}
\newcommand{\bfY}{\mathbf Y}
\newcommand{\ff}{\mathbf f}
\newcommand{\f}{\mathbb F}
\newcommand{\x}{\mathbf x}
A quantum CSS code \cite{CS96,Ste96} of length $n$ is defined by two binary matrices $H_X$ and $H_Z$,
each with $n$ columns, and such their row-spaces $W_X$ and $W_Z$ are orthogonal. Equivalently, $H_Z^TH_X=0$.
The matrices $H_X$ and $H_Z$ can be thought of as the parity-check matrices of
classical codes, $C_X=W_X^\perp$ and $C_Z=W_Z^\perp$ respectively.
The dimension of the quantum code is given by $n-\dim W_X - \dim W_Z$, equivalently
it is the dimension of either of the quotient spaces $C_X/W_Z$ or $C_Z/W_X$.
The Hamming distance $d_X$ (respectively $d_Z$) is defined as the smallest weight of a vector of
$C_X$ not in $W_Z$ (respectively $C_Z$ not in $W_X$). The minimum distance $d$
of the quantum code is defined as $d=\min(d_X,d_Z)$.
A quantum CSS code is said to be Low Density Parity Check (LDPC) if both
matrices $H_X$ and $H_Z$ have row and column weights bounded from above by a
constant.

If we define $X_0,X_1,X_2$ as the sets of rows of $H_X$,
columns of $H_X$ (or equivalently of $H_Z$) and rows of $H_Z$ respectively, then $H_Z^T$ and $H_X$ are the
matrices representing two linear maps $\partial_2$ and $\partial_1$
\[
\f_2^{X_2}\xrightarrow{\partial_{2}} \f_2^{X_1}\xrightarrow{\partial_{1}}
\f_2^{X_0}
\]
such that $\partial_1\partial_2=0$.
This process can be inverted, starting from a based chain complex of length $2$ we obtain the data of $H_X,~H_Z$ and hence of the CSS code. Thus, given any based chain complex $\X=(X_0,\ldots, X_d)$ with $d\geq 2,$ taking any $2-$subcomplex of the form
\[\f_2^{X_{k+1}}\xrightarrow{\partial_{k+1}} \f_2^{X_k}\xrightarrow{\partial_{k}}
\f_2^{X_{k-1}},\] gives rise to a CSS code. Unwinding the definitions, it is straight forward that \[d_Z=\text{CoSys}^{k}(\X),~~d_X=\text{Sys}_k(\X),\]
and that the dimension of the code is the $k-$th Betty number $h_k(\X).$
In this sense a pair of $k-$th homology and cohomology of a complex gives rise to a quantum CSS code. When $\X$ is local the code is LDPC.
The prototypical example of a LDPC quantum CSS code is the influential work \cite{Ki}.
\subsection{K\"unneth Theorem}
There is a natural and well known notion of tensor product of chain complexes (see e.g. \cite{AC}), which we now describe in the based case. Let
$\X$, $\Y$ be two based complexes. The tensor product $(\X\otimes \Y)$, is the based complex given by
\begin{equation*}
(\X\otimes \Y)_{k}=\bigsqcup_{i=0}^{k}X_{i}\times Y_{k-i}
\end{equation*}
hence the chain spaces are direct sums of tensor products
\begin{equation}\label{eq:decomp_chain}
C^{k}(\X\otimes \Y)=\oplus_{i=0}^{k}C^{i}(\X)\otimes C^{k-i}(\Y).
\end{equation}
The boundary and coboundary maps act on the element $v\otimes u\in C^{i}(\X)\otimes
C^{k-i}(\Y)$, by
\[
\partial_{k}^{\X\otimes \Y}(v\otimes u)=\partial_{i}^{\X}(v)\otimes
u+v\otimes \partial_{k-i}^{\Y}(u),\]\[\delta_{k}^{\X\otimes \Y}(v\otimes u)=\delta_{i}^{\X}(v)\otimes
u+v\otimes \delta_{k-i}^{\Y}(u),\]
and are extended linearly to $C^k(\X\otimes\Y).$
In other words, \[\delta^{\X\otimes \Y}=\delta^{\X}\otimes Id_\Y+Id_\X\otimes\delta^{Y},\] and similarly for the boundary map.
By slight abuse of notations we write $\delta^{\X}$ instead of $\delta^{\X}\otimes Id_\Y$ and $\delta^{\Y}$ instead of $Id_\X\otimes\delta^{Y}.$ Note that
\[\delta^{\X}(C^{i,k-i}(\X\otimes\Y))\subseteq C^{i+1,k-i}(\X\otimes\Y),~~\delta^{\Y}(C^{i,k-i}(\X\otimes\Y))\subseteq C^{i,k-i+1}(\X\otimes\Y).\]
It is also simple to see that if $\X$ is $Q$ local and $\Y$ is $Q'$ local, then their tensor product is $Q+Q'$ local.


Observe that\[|(\X\otimes \Y)_{k}|=\sum_{i=0}^{k}|X_{i}||Y_{k-i}|=\sum_{i=0}^{k}|(\X\otimes \Y)_{i,k-i}|,\]
where $(\X\otimes \Y)_{i,k-i}$ denotes $X_i\times Y_{k-i}.$ We also write $C^{i,k-i}(\X\otimes\Y)=C^{i}(\X)\otimes C^{k-i}(\Y),$ and we think of the vector space both as a subspace of $(\X\otimes\Y)_k,$ and as a direct summand of it. For a chain or cochain $\alpha$ we write $\alpha_{i,k-i}$ for its $C^{i,k-i}(\X\otimes\Y)-$component.

It is useful to identify the vector space $C^i(\X)\otimes C^{j-i}(\Y)$ with the space of $|\X_i|\times |\Y_{j-i}|$ matrices, with rows which correspond to $(j-i)-$cochains of $\Y$ and columns to $i-$cochains of $\X.$ With this identifications the $\X-$differential acts on columns while the $\Y-$differential on rows. We shall use this identification throughout the paper without further notice.

A simple observation which motivates our introduction of polyhedral complexes is this family of complexes is closed under taking tensor products, and that tensor products of simplicial complexes are polyhedral complexes.

The well known K\"unneth theorem relates the homology and cohomology of $\X\otimes\Y$ with those of $\X,\Y$
\begin{thm}\label{thm:Kunneth}
Let $\X$ and $\Y$ be complexes which are tensor products of simplicial complexes. The decomposition \eqref{eq:decomp_chain} induces isomorphisms
\[H_{k}(\X\otimes \Y)\cong\oplus_{i=0}^{k}H_{i}(\X)\otimes H_{k-i}(\Y)
\]
\[H_{k}(\X\otimes \Y)\cong\oplus_{i=0}^{k}H_{i}(\X)\otimes H_{k-i}(\Y)
\]
In particular, it holds that
\[h_k(\X\otimes\Y)=\sum_{i=0}^{k}h_i(\X)\times h_{k-i}(\Y).\]
Thus, if for some $i,~H^{i}(\X)$ and $H^{{k-i}}(\Y)$ are non trivial, then also the $k-$th homology and cohomology of $\X\otimes\Y$ are non trivial.
\end{thm}
The proof can be found in any standard text on algebraic topology. For example, the case of homology appears in \cite[Section~3.B]{Hat}, while the case of cohomology is handled in \cite[Theorem~3.15]{Hat}. In both cases it is proven for singular homologies and cohomologies, but these agree with polyhedral homology and cohomology when the polyhedral chain complex is a tensor product of simplicial complexes.

\subsection{Ramanujan complexes}
Ramanujan graphs were discovered in the seminal work \cite{LPS}. They have remarkable properties and, in particular, they are asymptotically the best possible expanders, at least spectrally.
The Ramanujan complexes, which were constructed in \cite{LSV1,LSV2} as quotients of Bruhat-Tits buildings, are a high dimensional generalization of Ramanujan graphs.
They are a family of high dimensional, bounded degree, simplicial complexes, and, like their one dimensional relatives, they also have many desired properties. In particular, the work of \cite{KKL,EK} has shown that the $d$-skeletons of Ramanujan complexes of dimension $d+1$ are explicit bounded degree cosystolic expanders, and they possess Gromov's topological overlapping property \cite{Gromov}.

The theory of Ramanujan complexes develops rapidly, and we refer the interested reader to the mentioned papers for precise definitions and a more detailed study.

In the following theorem we summarize the main properties of Ramanujan complexes that we shall need for this work, as were found by \cite{LSV1,LSV2,KKL,EK,EKZ}. In Lemma \ref{lem:polylog_systoles} we will need some additional properties, mainly of Bruhat-Tits buildings, but we will list them in the course of the proof. Additional properties of spherical buildings, that will be required for the proof of Theorem \ref{thm:strong_cofilling_spherical} will be mentioned in Appendix \ref{app:spherical}.

\begin{thm}[Known Properties of Ramanujan Complexes Theorem]\label{thm:LSV and their properties}
There exists an \emph{explicit} infinite family of bounded degree cosystolic expanders of every dimension $d$, where
\begin{itemize}
\item The number of vertices of a complex in the family (denoted by $n$) grows to infinity.
\item The local degrees are upper bounded by a constant $Q >0,$ that is independent of $n.$
\item All proper links are spherical buildings.
\item The $i$-th cohomology for $i=1,2$ do not vanish.
\item There is a linear lower bound for the length of the $i$-cosystole, if it exists, for every $i < d$.
\item The injectivity radius of the complex is $\Theta(\log n)$.
\end{itemize}
\end{thm}

\section{Products of cosystolic expanders}\label{sec:prod}
We now restate and prove the first main theorem, Theorem \ref{thm:1} from the introduction, in greater generality.
\begin{thm}\label{thm:linear_cosys}
Let $\X,\Y$ be complexes of dimensions $d,d',$ respectively, weighted by the Hamming norm. Assume that
\begin{itemize}
\item $\X$ has cosystolic bounds $\eta_i,~i\leq d-1.$
\item $\Y$ has upper localities $q'_i,~i\leq d'.$
\item $\X$ has cofilling constants $\mu_i,~i\leq d-1$.
\end{itemize}
\begin{enumerate}
\item If $\Y$ has cosystolic bounds $\eta'_i,~i\leq d-1,$
then there exist constants $\lambda_l,$ for $l\leq \min\{d,d'\},$ which depend only on \[(\eta_i)_{i\leq l},~(\eta'_i)_{i\leq l},~(q'_i)_{i\leq l},~(\mu_i)_{i\leq l},\]
such that for all $l<\min\{d,d'\},$ if the $l$-th cohomology is non zero, then \[\text{CoSys}^{l}(\X\otimes \Y)\geq \lambda_l\min_{i\leq l}\{|X_{l-i}||Y_i|\}.\]
\item If $\Y$ has the collective cofilling property with constants $(\mus_i)',~i\leq d'-1$ then $\X\times\Y$ has the cofilling property, for all $l<\min\{d,d'\}$ with cofilling constants $\nu_l,$ which depend only on \[(\eta_i)_{i\leq l},~(q'_i)_{i\leq l},~(\mu_i)_{i\leq l},~((\mus_i)')_{i\leq l}.\]
\item If both $\X$ and $\Y$ have the collective cofilling property with constants $\mus_i,~i\leq d-1$ and $(\mus_i)',~i\leq d'-1$ respectively, then the product has the collective cofilling property for all
$l<\min\{d,d'\}$ with collective cofilling constants $\nu_l^{\text{coll}},$ which depend only on \[(\eta_i)_{i\leq l},~(q'_i)_{i\leq l},~(\mus_i)_{i\leq l},~((\mus_i)')_{i\leq l}.\]
\end{enumerate}
\end{thm}
\begin{proof}
Let $n_i=|X_i|,~n'_i=|Y_i|.$
Then for $l\leq \min\{d,d'\},~~|(\X\otimes\Y)_{i,l-i}|=n_in'_{l-i}.$
Write \[\lambda_l =
\frac{\prod_{i=0}^{l-1}\eta_{i}}{\prod_{i=0}^{l-1} q'_i\mu_{l-i-1}
}\min_{j\leq l}\{\eta_{l-j}\eta'_j\}\text{ and }N:=\min_{i\leq l} n_in'_{l-i}.
\]
Note that $\eta_\bullet,\eta'_\bullet\leq 1$ while $\mu_\bullet,q'_\bullet\geq1.$

We begin with the first case.
Let $\alpha$ be a short cocycle, that is a cocycle of norm
\[|\alpha|< \lambda_lN.\]
We will construct a sequence $\alpha_j,~j=0,\ldots,l+1$ with the following properties:
\begin{itemize}
\item $\alpha_0=\alpha.$
\item For all $j$ and for all $i<j,~~(\alpha_j)_{l-i,i}=0,$ where $(\alpha_j)_{l-i,i}$ is the $(l-i,i)-$component of $\alpha_j,$ and for $i>j,~~(\alpha_j)_{l-i,i}=\alpha_{l-i,i}.$
\item\label{it:ineq}
$|\alpha_j|\leq b_j |\alpha|<b_j\lambda_l|N|$ where
\begin{equation}\label{eq:b_j_def}b_0=1,~~b_1=
q'_0\mu_{l-1},~~ b_{j+1} =  \frac{q'_j\mu_{l-j-1}}{\eta_{l-j}}b_j,~~j\geq 1. \end{equation}
\item $\alpha_{j+1}-\alpha_j$ is a coboundary. 
\end{itemize}
The existence of this sequence implies, in particular, that $\alpha$ is a coboundary, since $\alpha_{l+1}=0$ and differs from $\alpha$ by a coboundary. In turn this implies that all $l-$cocycles of $\X\otimes\Y$ which are not coboundaries are of norm at least $\lambda_lN,$ proving the theorem.

Observe that the definition of the bounds $b_j,$ and simple induction, imply
\begin{equation}\label{eq:b_j_ineq}
b_j\lambda_l\leq \min_{i\leq l}\{\eta_{l-i}\eta'_i\},
\end{equation}
indeed, simple induction shows
\[b_j\lambda_l\leq\frac{\prod_{i=0}^{l-j}\eta_{i}}{\prod_{i=j}^{l-1} q'_i\mu_{l-i-1}}
\min_j\{\eta_{l-j}\eta'_j\}.
\]
Note also that $b_j$ is independent of the cosystolic bounds of $\Y.$

For $j=0,\ldots, l$ let $\bar{e}^{j}_i,~~i\in[h^{j}(\X)]$ and $\bar{f}^{j}_i,~~i\in[h^{j}(\Y)]$ be bases for the $j^{th}$ cohomology of $H^{j}(\X),~H^{j}(\Y)$ respectively. Let $e^{j}_i\in Z^{j}(\X)$ be a representative of $\bar{e}^{j}_i$ and define $f^{j}_i$ similarly.

By Theorem~\ref{thm:Kunneth} $\alpha$ can be written as
\[\delta\beta + \sum_{(j,i,i')\in I} e^{l-j}_i\otimes f^{j}_{i'},\]
for some set $I\subseteq \bigsqcup_{j=0}^{l}[h^{l-j}(\X)]\times[h^{j}(\Y)]$ and a cochain $\beta\in C^{l-1}(\X\otimes \Y).$
We consider $e^{j}_i\otimes f^{l-j}_{i'}$ as a cochain in $C^{l-j,j}(\X\otimes \Y)\hookrightarrow C^l(\X\otimes \Y).$ Write $I_j\subseteq I$ for the subset of triples with first entry $j.$

We begin with constructing $\alpha_1.$ Since \[(\delta\alpha)_{l+1,0}=\delta^{\X}\alpha_{l,0}=0,\]
all columns of $\alpha_{l,0}$ are $l-$cocycles of $\X$. Moreover,
\begin{equation}\label{eq:alpha_1-first}\alpha_{l,0}=\delta^{\X}(\beta_{l-1,0}) + \sum_{(0,i,i')\in I_0} e^{l}_i\otimes f^{0}_{i'}.\end{equation} We first show that $I_0=\emptyset.$ We rewrite \eqref{eq:alpha_1-first} as
\begin{equation}\label{alpha_1_second}\alpha_{l,0}=\delta^{\X}(\beta_{l-1,0}) + \sum_{i\in[h^{l}(\X)]} e^{l}_i\otimes v_i,\end{equation}
where $v_i\in Z^{0}(\Y)$ a linear combination of the $f^{0}_{i'}$'s. As the elements $f^{0}_{i'}$ are linearly independent, $I_0\neq\emptyset$ precisely if at least one $v_i$ is non zero.
In this case the columns indexed by \[S_0=\bigcup_i \text{supp}(v_i)\] are $\X-$cocycles which are not coboundaries. Indeed,
the column $t\in S_0$ equals
\[\sum_{i: t\in \text{supp}(v_i)}e^{l}_i + \delta^{\X} \beta^{t}_{l-1,0}, \]
where $\beta^{t}_{l-1,0}$ is the $t^{th}-$column of $\beta_{l-1,0}.$ This combination is a cocycle, but not a coboundary, since $\bar{e}^{l}_i$ are linearly independent and the summation over $i$ is non vacuous as $t\in S_0.$ Each column in $S_0$ is of $\X-$norm at least $\eta_ln_l.$ If some $v_i\neq 0,$ then since $v_i$ itself is a $0-\Y-$cocycle which is not a coboundary then its $\Y-$norm is at least $\eta'_0n'_0,$ hence $S_0$ is of size at least $\eta'_0n'_0.$
Putting together,
\[\eta_l\eta'_0n_ln'_0\leq |\alpha_{l,0}|\leq|\alpha|<\lambda_lN=b_0\lambda_lN,\]
which contradicts \eqref{eq:b_j_ineq} applied to $b_0\lambda_l.$
\
Therefore $\alpha_{l,0}=\delta^{\X}(\beta_{l-1,0}),$ and in particular all its columns are $\X-$coboundaries. By the definition of $\mu_l,$ we can find a cochain $\gamma_0\in C^{l-1}(\X)\otimes C^{0}(\Y)$ with
\begin{equation}\label{eq:gamma_0}
\alpha_{l,0}=\delta^{\X}\gamma_0,~~~|\gamma_0|\leq\mu_{l-1}|\alpha_{l,0}|,\end{equation}
this is done by taking a short $\delta^{\X}-$preimage for each column of $\alpha_{l,0}.$
Define
\[\alpha_1=\alpha+\delta^{\X\otimes\Y}\gamma_0=\alpha-\alpha_{l,0}+\delta^{\Y}\gamma_0.\]
$\alpha_1$ satisfies the second and forth requirements.
For the third one,
\[|\alpha_1|\leq |\alpha-\alpha_{l,0}|+|\delta^{\Y}\gamma_0| = |\alpha|-|\alpha_{l,0}|+|\delta^{\Y}\gamma_0|\leq |\alpha|+(q'_0\mu_{l-1}-1)|\alpha_{l,0}|\leq q'_0\mu_{l-1}|\alpha|=b_1|\alpha|,\]
where we have used the definition of $b_1$ and the observation
\[|\delta^{\Y}\gamma_0|\leq q'_0|\gamma_0|\leq q'_0\mu_{l-1}|\alpha_{l,0}|.\]

Write $\beta_1=\beta-\gamma_0.$

We will construct $\alpha_j,~j>1$ inductively.
We will construct, alongside with $\alpha_j$ an additional class $\beta_j\in C^{l-1}(\X\otimes \Y),~~j\geq1$ which satisfies properties we now list.
\begin{itemize}
\item $(\beta_j)_{(l-1-s,s)}=0,$ for $s<j-1.$
\item \begin{equation}\label{eq:alp_bet}
\alpha_j=\delta \beta_j + \sum_{(s,i,i')\in I} e^{l-s}_i\otimes f^{s}_{i'},\end{equation}
with the same set $I\subseteq \bigsqcup_{s=0}^{l}[h^{l-s}(\X)]\times[h^{s}(\Y)].$
\end{itemize}
We will also show inductively, in parallel to the construction of $\alpha_j,\beta_j,$ that $I_s=\emptyset,$ for all $s<j.$
Note that $\beta_1$ satisfies the required properties with respect to $\alpha_1,$ and that $I_0$ is indeed empty, so that the basis of induction is satisfied.

Suppose we have constructed $\alpha_j,\beta_j,$ and that we have shown $I_s=\emptyset$ for $s\leq j-1.$
Note that
\begin{equation}\label{eq:d_x_beta}
(\delta\beta_j)_{l-j+1,j-1}=(\alpha_{j})_{l-j+1,j-1}=0\Rightarrow\delta^{\X}\left((\beta_j)_{(l-j,j-1)}\right)=0.
\end{equation}
We now construct $\alpha_{j+1},\beta_{j+1}.$ By construction,
\[(\alpha_j)_{l-i,i}=0,~~i=0,\ldots,j-1.\]
We first show that $I_j=\emptyset.$ We can write, by restricting to the $(l-j,j)-$component of $\alpha_j.$
\begin{equation}\label{eq:alp_j_decomp}(\alpha_j)_{l-j,j}=\sum_{i\in[h^{l-j}(\X)]}e^{l-j}_i\otimes v_i + \delta^{\Y}\beta'_j + \delta^{\X}\beta''_j, \end{equation}
where $\beta'_j=(\beta_j)_{l-j,j-1}\in C^{l-j,j-1}(\X\otimes \Y),~~\beta''_j=(\beta_j)_{l-j-1,j}\in C^{l-j-1,j}(\X\otimes \Y)$ and, by \eqref{eq:d_x_beta},
\[\delta^{\X}\beta'_j=0,\]
so that all columns of $\beta'_j$ are $\X-$cocycles. As above
\[v_i=\sum_{i':(j,i,i')\in I_j}f_{i'}^{j}.\]
$I_j\neq\emptyset$ precisely if at least one $v_i$ in the summation is non zero. Any such $v_i\neq 0$ must be a cocycle which is not coboundary.
Since $\delta^{\X}\beta'_j=0$ we can write $\beta'_j=\beta^Z_{j}+\beta^B_{j}$ where
\begin{equation}\label{eq:beta_z_beta_b}\beta^Z_{j} = \sum_{i\in[h^{l-j}(\X)]}e^{l-j}_i\otimes u_i,~~\beta^B_{j}=\sum_i (\delta^{\X} g_i)\otimes w_i,\end{equation}for arbitrary elements $u_i,~w_i\in C^{j-1}(\Y)$\footnote{$u_i$ depend on $j$ and hence should be written $u_i^j,$ and the same comment holds for $v_i$ and $w_i,$ and $g_i.$ In order to lighten notations we omit the $j-$dependence for the notations, but we will need to add the superscript $j$ later, when we consider cofilling problems.} and $g_i\in C^{l-j-1}(\X).$
Write
\[v'_i=v_i+\delta^{\Y}u_i.\] Then whenever $v_i\neq 0$ also $v'_i\neq 0$ and in this case both are $\Y-$cocycles which are not coboundaries.
Let \begin{equation}\label{eq:S_j}S_j=\bigcup_i \text{supp}(v'_i).\end{equation} Since we assumed $I_j\neq\emptyset,$ then also $S_j\neq \emptyset.$ Moreover, the same reasoning we used to lower bound $|S_0|$ shows,
\[\eta'_jn'_j\leq|S_j|.\]
The same reasoning also shows that the columns of \[\sum_{i\in[h^{l-j}(\X)]} e^{l-j}_i\otimes v_i+\delta^{\Y}\beta^Z_{j}=\sum_{i\in[h^{l-j}(\X)]} e^{l-j}_i\otimes v'_i\] indexed by $S_j$
are $\X-$cocycles which are not coboundaries.
All columns of
\[\delta^{\X}\beta''_j\] are clearly $\X-$coboundaries.
Since
\begin{equation}\label{eq:beta-B}
\delta^{\Y}\beta^{B}_j=\delta^{\Y}(\sum_i\delta^{\X}g_i\otimes w_i)=\delta^{\X}(\sum_i g_i\otimes (\delta^{\Y}w_i)),
\end{equation} also all columns of $\delta^{\Y}\beta^B_{j}$ are $\X-$coboundaries.

Thus, by \eqref{eq:alp_j_decomp}, the columns of $(\alpha_j)_{l-j,j}$ indexed by $S_j$ are all cocycles which are not coboundaries, and therefore are of $\X-$norm at least $\eta_{l-j}n_{l-j}.$
Putting together, we see that if $I_j\neq \emptyset$ then
\[\eta_{l-j}\eta'_jn_{l-j}n'_{j}\leq |(\alpha_j)_{l-j,j}|\leq|\alpha_j|< b_j\lambda_lN,\]
which contradicts \eqref{eq:b_j_ineq}.

So we see that $I_j=\emptyset,$ hence all $v_i=0.$ In this case, using the same notations,
\begin{equation}\label{eq:alpha_j_l-j_j}(\alpha_j)_{l-j,j}= \delta^{\Y}\beta^Z_{j} +\delta^{\Y}\beta^B_{j} + \delta^{\X}\beta''_j.\end{equation}
We can still define $S_j$ to be $\bigcup_i \text{supp}(v'_i)=\bigcup_i\text{supp}(\delta^{\Y} u_i),$ and again, by the same consideration, all columns of $(\alpha_j)_{l-j,j}$ indexed by $S_j$ are of $\X-$norm at least $\eta_{l-j}n_{l-j}.$
Write \begin{equation}\label{eq:alpha'}\alpha'_j=\alpha_j+\delta\beta^Z_{j}=\alpha_j+\delta^{\Y}\beta^Z_{j}.\end{equation} $\alpha'_j$ differs from $\alpha_j$ only in the columns labeled by $S_j,$ hence
\begin{equation}\label{eq:ineq_alpha'}|\alpha'_j|\leq 
\frac{1}{\eta_{l-j}}|\alpha_j|.\end{equation}
In addition $\alpha'_j$ differs from $\alpha_j$ by a coboundary, and all columns of $(\alpha'_j)_{(l-j,j)}$ are $\X-$coboundaries.
In fact,
\begin{equation}\label{eq:alpha'_coumns}
(\alpha'_j)_{(l-j,j)}=\delta^{\X}\beta''_j+\delta^{\Y}\beta^{B}_j=\delta^{\X}\beta''_j+\delta^{\X}(\sum_ig_i\otimes (\delta^{\Y}w_i)),
\end{equation}
where we have used \eqref{eq:beta-B}.
Therefore, as in the first step, we can find a short $\delta^{\X}-$preimage for each column of $(\alpha'_j)_{l-j,j}$ and this way to define
\begin{equation}
\label{eq:gamma_j}\gamma_j\in C^{l-j-1,j}(\X\otimes \Y),~~\delta^{\X}\gamma_j=(\alpha'_j)_{l-j,j},~~|\gamma_j|\leq \mu_{l-j-1}|(\alpha'_j)_{l-j,j}|\leq \mu_{l-j-1}|\alpha'_j|.\end{equation}
Put
\begin{equation}\label{eq:alpha_j+1}\alpha_{j+1}=\alpha'_j+\delta\gamma_j,\end{equation}and again it satisfies the second and forth requirements. Regarding the third property:
\[\alpha_j=\alpha'_j+\delta^\X\gamma_j +\delta^\Y\gamma_j=\alpha'_j-(\alpha'_j)_{l-j,j}+\delta^\Y\gamma_j,\]
thus,
\begin{align*}|\alpha_j|=|\alpha'_j|-|(\alpha'_j)_{l-j,j}|+|\delta\gamma_j|&\leq |\alpha'_j|-|(\alpha'_j)_{l-j,j}|+ q'_j\mu_{l-j-1}|(\alpha'_j)_{l-j,j}|\\&\leq |\alpha'_j|+(q'_j\mu_{l-j-1}-1)|(\alpha'_j)_{l-j,j}|\leq\frac{q'_j\mu_{l-j-1}}{\eta_{l-j}}|\alpha_j|\leq b_{j+1}|\alpha|,\end{align*}
where we have used \eqref{eq:gamma_j}, \eqref{eq:ineq_alpha'}, \eqref{eq:b_j_def} and the induction hypothesis.

We also set
\begin{equation}\label{eq:beta_j+1}\beta_{j+1}=\beta_j-\beta'_j-\gamma_j-\sum_i g_i\otimes (\delta^{\Y}w_i).\end{equation} We claim it satisfies the requirements.
The first property holds since the $(l-j,j-1)-$component of $\beta_j$ is $\beta'_j,$ and the other corrections belong to $C^{(l-j-1,j)}(\X\otimes\Y).$
For the second one, note that since $I_j=\emptyset,$ it is enough to show
\[\delta(\beta_{j+1}-\beta_j)=\alpha_{j+1}-\alpha_j.\]
From the definitions
\[\alpha_{j+1}-\alpha_j=\delta(\beta^{Z}_j+\gamma_j).\]
Now, \eqref{eq:beta_j+1} and the definition of $\beta^{Z}_j,~\beta^{B}_j$ give
\[\delta(\beta_{j+1}-\beta_j)=\delta(\beta^{Z}_j+\beta^{B}_j+\sum_i g_i\otimes (\delta^{\Y}w_i)+\gamma_j)=\delta(\beta^{Z}_j+\gamma_j),\]
where we have used \eqref{eq:beta-B} to cancel the middle terms in the middle expression.
The induction follows.

We now turn to prove the second claim, that the product has the cofilling property if $\Y$ has the collective cofilling property, as in the statement of the theorem.

We continue with the above notations, only that this time $\alpha\in B^l(\X\times\Y)$ is an arbitrary coboundary, without the smallness assumption, and hence also the estimate \eqref{it:ineq} will not hold. We will find $\tbeta\in C^{l-1}(\X\times\Y)$ with $\delta^l\tbeta=\alpha.$
In this case, from Theorem \ref{thm:Kunneth}, $I=\emptyset.$\footnote{In the previous part, in order to prove that $I=\emptyset,$ we had to use the cosystolic bounds of $\Y.$}
Using \eqref{eq:alpha_j_l-j_j} and \eqref{eq:beta_z_beta_b} we can write
\[
(\alpha_j)_{l-j,j}=  \sum_{i\in[h^{l-j}(\X)]}e^{l-j}_i\otimes \delta^{\Y}u^j_i
+\delta^{\Y}\beta^B_{j} + \delta^{\X}\beta''_j.
\]
Let $\widetilde{u}^j_i\in C^{j-1}(\Y)$ be $\delta^{\Y}$-preimages of $\delta^\Y(u^j_i),~i\in[h^{l-j}(\X)],$ which satisfy
\begin{equation}\label{eq:support_of_tilde_us_union}
|\bigcup_i \widetilde{u}^j_i|\leq (\mus_j)'|\bigcup_i \delta^{\Y}u^j_i|.
\end{equation}
The existence of such elements is guaranteed by the collective cofilling property of $\Y.$
Let $\tbeta\in C^{l-1}(\X\times\Y)$ be the $(l-1)$-cochain
\[\tbeta_{l-1-j,j}=\gamma_j+\sum_{i\in[h^{l-1-j}(\X)]}e^{l-1-j}_i\otimes \widetilde{u}^{1+j}_i, ~~0\leq j\leq l-1.\]
We first show that $\delta\tbeta=\alpha$ and then estimate $|\tbeta|.$

\[(\delta\tbeta)_{l,0}=\delta^\X\gamma_0+\sum_{i\in[h^{l-1}(\X)]}\delta^\X(e^{l-1}_i)\otimes \widetilde{u}^1_i=\delta^\X\gamma_0=\alpha_{l,0},\]
from the definition of $\gamma_0$ and the fact that $e^{l-j}_i\in Z^{l}(\X).$

For $j>0$ we need to show that $\alpha_{l-j,j}$ equals
\begin{equation}\label{eq:delta_tbeta_l-j_j}
(\delta\tbeta)_{l-j,j}=\delta^\Y\tbeta_{l-j,j-1}+\delta^\X\tbeta_{l-1-j,j}=
\delta^\Y\gamma_{j-1}+\sum_{i\in[h^{l-j}(\X)]}e^{l-j}_i\otimes \delta^\Y\widetilde{u}^{j}_i+\delta^\X\gamma_j,
\end{equation}
where we set $\gamma_l=0$ and used the $\delta^\X$-closeness of $e^{l-1-j}_i.$

Consider $\alpha_{j-1}.$ Its $(l-j,j)-$component equals $\alpha_{l-j,j},$ by construction. Similarly, by construction, the $(l-j,j)-$component of $\alpha_{j+1}$ is 0. Thus, \[\alpha_{l-j,j}=(\alpha_{j+1}-\alpha_{j-1})_{l-j,j}.\]
On the other hand, using equations \eqref{eq:alpha'},~\eqref{eq:alpha_j+1}, we have
\[\alpha_{j+1}-\alpha_{j-1}=\delta\gamma_{j-1}+\delta\beta^Z_{j-1}+
\delta\gamma_{j}+\delta\beta^Z_{j}.\] Taking the $(l-j,j)-$component, using \eqref{eq:beta_z_beta_b} and $\delta^{\Y}\widetilde{u}^j_i=\delta^{\Y}{u}^j_i$ proves \eqref{eq:delta_tbeta_l-j_j}.

In order to estimate $\tbeta,$ note that
\[|\sum_{i\in[h^{l-j}(\X)]}e^{l-j}_i\otimes \widetilde{u}^j_i|\leq n_{l-j}|\bigcup\widetilde{u}^j_i|\leq (\mus_j)'n_{l-j}|\delta^{\Y}{u}^j_i|=(\mus_j)'n_{l-j}|S_j|,\]where we have used \eqref{eq:support_of_tilde_us_union} and the definition of $S_j,~$\eqref{eq:S_j}.
On the other hand, as we saw above, each column of $\alpha_j$ indexed by $S_j,$ is a non trivial cocycle, hence of length at least $\eta_{l-j}n_{l-j}.$
Thus,
\begin{equation}\label{eq:tbeta_1st_estimate}|\alpha_j|\geq |(\alpha_{j})_{l-j,j}|\geq \eta_{l-j}n_{l-j}|S_j|\geq
\frac{\eta_{l-j}}{(\mus_j)'}|\sum_{i\in[h^{l-j}(\X)]}e^{l-j}_i\otimes \widetilde{u}_i|.\end{equation}
Although we gave up on the smallness assumption for $\alpha,$ \eqref{eq:gamma_0},~\eqref{eq:ineq_alpha'},~\eqref{eq:gamma_j} and $\alpha_j\leq b_j|\alpha|$ still hold, with $b_j$ defined in \eqref{eq:b_j_def}, since deriving them only used $I=\emptyset,$ which holds in our case. Therefore
\[|\sum_{i\in[h^{l-j}(\X)]}e^{l-j}_i\otimes \widetilde{u}^j_i|\leq
\frac{(\mus_j)'}{\eta_{l-j}}|\alpha_j|\leq \frac{(\mus_j)'b_j}{\eta_{l-j}}|\alpha|\]
and
\[|\gamma_0|\leq\mu_{l-1}|\alpha|,~~|\gamma_j|\leq \frac{\mu_{l-j-1}}{\eta_{l-j}}|\alpha_j|\leq
\frac{\mu_{l-j-1}b_j}{\eta_{l-j}}|\alpha|,~~j>0.\]
Summing these equations we see that
\[|\tbeta_{l-1,0}|\leq (\mu_{l-1} + \frac{(\mus_{1+j})'b_{1+j}}{\eta_{l-1-j}}|
)|\alpha|,~~
|\tbeta_{l-j-1,j}|\leq (\frac{\mu_{l-j-1}b_j}{\eta_{l-j}} + \frac{(\mus_{1+j})'b_{1+j}}{\eta_{l-1-j}}|
)|\alpha|,~~j>0
,\]
and hence
\[|\tbeta|\leq \nu_l|\alpha|,\]
for \begin{equation}\label{eq:nu_l}\nu_l=
\mu_{l-1}+\sum_{j=1}^{l-1}\frac{b_j}{\eta_{l-j}}((\mus_j)'+\mu_{l-j-1}).\end{equation}
As claimed.

We now turn to the last item of the theorem, concerning the product of two complexes having the collective cofilling property.
This time we begin with coboundaries $\alpha_{a}\in B^{l}(\X\times\Y),~a=1,\ldots,m,$ and we aim to construct $(l-1)-$cochains, $\tbeta_a\in C^{l-1}(\X\times\Y),~a=1,\ldots,m$ such that for all $a,~\delta\tbeta_a=\alpha_a$ and
\[|\bigcup_{a\in [m]}\tbeta_a|\leq\nu_l|\bigcup_{a\in [m]}\alpha_a|.\]The constant $\nu_l$ will be given by the same expression as in the previous case, only with each $\mu_i$ replaced by $\mu'_i.$
The proof is very similar to the previous case, so we only describe the necessary changes.
We use the same notations as before, only we add with an additional index $a,$ so that we have classes $(\gamma_a)_j,~(u_a)^j_i,~a\in[m],$ etc., playing the analogous role to before, only with respect to $\alpha_a.$
There are precisely two differences in the constructions. The first is that instead of choosing $(\gamma_a)_j$ by taking short $\delta^\X-$preimages for the columns of $(\alpha'_a)_j,$ we choose the preimages $(\gamma_a)_j,~a\in[m]$ so that for any $t\in Y_j,$ the norm of the \emph{union} over $a\in[m]$ of columns indexed by $t$ in $(\gamma_a)_j,$ is at most $\mus_{l-j-1}$ times the norm of the \emph{union} over $a\in[m]$ of the columns indexed $t$ in $(\alpha'_a)_j.$ This is achieved by using the collective cofilling of $\X.$
Similarly, instead of choosing $(\widetilde{u}_a)^j_i$ as $\delta^{\Y}-$preimages of $\delta^{\Y}(u_a)^j_i$ satisfying \eqref{eq:support_of_tilde_us_union} for each $a$ separately, we choose them to satisfy
\[|\bigcup_{a,i} (\widetilde{u}_a)^j_i|\leq (\mus_j)'|\bigcup_{a,i} \delta^{\Y}(u_a)^j_i|.\]
All the analysis remains the same, and the proof follows.

%
%
\end{proof}
\section{The collective cofilling property of Ramanujan complexes}\label{sec:strong_coff_LSV}
In this section we shall prove the collective cofilling for Ramanujan complexes. We will adapt the strategy of \cite{KKL,EK} to our situation. Our proof will rely strongly on the notion of local minimality for collections of cochains, which generalizes the corresponding notion for single chains, defined in \cite{KKL}. We will restrict our attention to simplicial complexes, and recall useful notions such as the link, localization to the link etc., following the notations of \cite{EK}. For notational convenience we augment $\X$ by a $-1$-layer $X_{-1}=\{\emptyset\},$ such that $\partial_0$ maps every vertex in $X_0$ to the single generator of $\F_2^{X_{-1}}.$ We write $\sigma\subset\tau$ for two simplices, if $\sigma$ is contained in $\tau.$ For two simplices $\sigma,\tau$ which share no vertex we write $\sigma\sqcup\tau$ for the simplex whose vertices are the union of the vertices of the two. Similarly $\sigma\cap\tau$ is the simplex whose vertices belong to both $\sigma$ and $\tau.$

Throughout the section it will be more convenient to work with the weight function \eqref{eq:lob_norm}, and we will work solely with this weight function. In particular, collective cofilling constants will be considered with respect to it.
\begin{definition}
The $r-$\emph{container} of $\alpha\in C^k(\X)$ is the cochain $\Gamma^r(\alpha)\in C^r(\X)$ defined by
\[\Gamma^r(\alpha)=\{\tau\in X_r\big|\exists\sigma\in\alpha~\text{s.t }\sigma\subseteq\tau\}.\]
\end{definition}
\cite[Lemma 2.3]{EK} says
\begin{equation}\label{eq:EK2.3}
\nr \alpha\nr \leq \nr \Gamma^r(\alpha)\nr \leq \binom{r+1}{k+1}\nr \alpha\nr.
\end{equation}

\begin{definition}\label{def:link,loc,glob}
Let $\X$ be a complex of dimension $d,~k\leq d-1,$ and $\sigma\in X_k.$
\begin{enumerate}
\item The \emph{link} $\X_\sigma$ of $\sigma$ is the $d-k-1-$simplicial complex whose $j$-simplices are $\tau\in X_j$ which share no vertex with $\sigma$ and $\sigma\sqcup\tau\in X_{k+j+1}.$\\$\delta_\sigma: C^*(\X_\sigma)\to C^{*+1}(\X_\sigma)$ denotes the coboundary map, and $\nr-\nr: C^*(\X_\sigma)\to[0,1]$ denotes the weight function \eqref{eq:lob_norm} with respect to $\X_\sigma$. A link is proper if it is not the link of $\emptyset$ (which is, by definition, $\X$).
\item The \emph{localization} with respect to $\sigma$ is the map
\[I_\sigma:C^*(\X)\to C^{*-k-1}(\X_{\sigma}),~~I_\sigma(\alpha)=\{\tau\in \X_\sigma\big|~\sigma\sqcup\tau\in\alpha\},\]
where as usual we identify a cochain with its support.
\item The \emph{lifting} map with respect to $\sigma$ is the map
\[I^\sigma:C^{*}(\X_{\sigma})\to C^{*+k+1}(\X),~~I^\sigma(\alpha)=\{\tau\sqcup\sigma\in \X\big|~\tau\in\alpha\}.\]
\end{enumerate}
\end{definition}
The described notions have many properties, some of which we state, and we omit all proofs which appear in \cite[Subsection 3.2]{EK}, and are straight forward.
It holds that
\begin{equation}\label{eq:EK2.12}
I^\sigma(I_\sigma(\alpha))=\{\tau\in\alpha\big|~\sigma\subset\tau\}.
\end{equation}
and for $\alpha\in C^*(\X_\sigma)$
\begin{equation}\label{eq:EK2.11}
I_\sigma(I^\sigma(\alpha))=\alpha.
\end{equation}
In addition we have
\begin{obs}\label{obs:EK2.7.1}
For any $\sigma\in X_j$ the following hold.
\begin{enumerate}
\item For any $\alpha\in C^{k-j-1}(\X_\sigma)$\[\nr I^\sigma(\alpha)\nr = \binom{k+j+2}{k+1}w(\sigma)\nr\alpha\nr_\sigma\]
\item For any $\alpha,\beta\in C^i(\X_\sigma)$
\[\nr\alpha\nr_\sigma\leq\nr\beta\nr_\sigma\Leftrightarrow \nr I^\sigma(\alpha)\nr\leq \nr I^\sigma(\beta)\nr.\]
\item For any $\alpha_1,\ldots,\alpha_m\in C^k(\X),~\beta_1,\ldots,\beta_m\in C^{k-j-1}(\X_\sigma),$ if $\nr \bigcup_a\alpha_a\nr\leq \nr \bigcup_a\left(\alpha_a+I^\sigma(\beta_a)\right)\nr$ then $\nr \bigcup_a I_\sigma(\alpha_a)\nr_\sigma\leq \nr \bigcup_a\left(I_\sigma(\alpha_a)+\beta_a\right)\nr.$
\end{enumerate}
\end{obs}
\begin{proof}
Only the third item is new, and is straight forward:
\begin{align*}
\nr\bigcup_a\alpha_a\nr-\nr I^\sigma(I_\sigma(\bigcup_a\alpha_a))\nr=&\sum_{\tau\in \bigcup_a\alpha_a,~\tau\nsupseteq\sigma}w(\tau)=
\\&=\sum_{\tau\in \bigcup_a(\alpha_a+I^\sigma(\beta_a)),~\tau\nsupseteq\sigma}w(\tau)=
\nr\bigcup_a(\alpha_a+I^\sigma(\beta_a))\nr-\nr\bigcup_a
I^\sigma(I_\sigma(\alpha_a+I^\sigma(\beta_a)))\nr.
\end{align*}
Since $\nr \bigcup_a\alpha_a\nr\leq \nr \bigcup_a\alpha_a+I^\sigma(\beta_a)\nr$ we deduce \[\nr I^\sigma(I_\sigma(\bigcup_a\alpha_a))\nr\leq \nr\bigcup_a
I^\sigma(I_\sigma(\alpha_a+I^\sigma(\beta_a)))\nr=
\nr
I^\sigma(I_\sigma(\bigcup_a(\alpha_a+I^\sigma(\beta_a))))\nr.\]
Using the second item and \eqref{eq:EK2.11} we obtain
\[\nr \bigcup_a I_\sigma(\alpha_a)\nr_\sigma =\nr  I_\sigma(\bigcup_a\alpha_a)\nr_\sigma\leq \nr I_\sigma(\bigcup_a(\alpha_a+I^\sigma(\beta_a)))\nr_\sigma=\nr\bigcup_a(I_\sigma(\alpha_a+I^\sigma(\beta_a)))\nr_\sigma
=\nr\bigcup_a(I_\sigma(\alpha_a)+\beta_a)\nr_\sigma.
\]
\end{proof}
\begin{definition}\label{def:minimality for collections}
Let $\X$ be a simplicial complex.
A collection of cochains $\alpha_1,\ldots,\alpha_m\in C^k(\X)$ is \emph{minimal} if for any $\gamma_1,\ldots,\gamma_m\in B^k(\X)$
\[\nr\bigcup_{a\in [m]}\alpha_a\nr\leq
\nr\bigcup_{a\in [m]}(\alpha_a+\gamma_a)\nr.\]

A collection of cochains $\alpha_1,\ldots,\alpha_m\in C^k(\X)$ is \emph{locally minimal} if for any simplex $\sigma,$ the localizations
\[I_\sigma(\alpha_1),\ldots,I_\sigma(\alpha_m)\]form a minimal collection with respect to the link of $\sigma.$
\end{definition}
\begin{lemma}\label{lem:restirction_of_minimal}
\begin{enumerate}
\item Every minimal collection is locally minimal.
\item If $\{\alpha_1,\ldots,\alpha_m\}$ is a minimal collection of $k$-cochains, and $\beta\in C^k(\X)$ is arbitrary then $\{\alpha_1\cap\beta,\ldots,\alpha_m\cap\beta\}$ is also a minimal collection.
\item Suppose in addition that the sizes of the all proper links of $\X$ are bounded from above by $Q$. Then for any collection $\{\alpha_1,\ldots,\alpha_m\}$ of $k$-cochains, there exist $(k-1)$-cochains $\{\gamma_1,\ldots,\gamma_m\}$ such that
    \[\nr\bigcup_{a\in[m]}\gamma_a\nr\leq Q\nr\bigcup_{a\in[m]}\alpha_a\nr,\]
    the collection $\{\alpha'_1,\ldots,\alpha'_m\}=\{\alpha_1+\delta\gamma_1,\ldots,\alpha_m+\delta\gamma_m\}$
    is locally minimal and satisfies
    \[\nr\bigcup_{a\in[m]}\alpha'_a\nr\leq \nr\bigcup_{a\in[m]}\alpha_a\nr.\]
\end{enumerate}
\end{lemma}
\begin{proof}
For the first item, assume $\valpha=\{\alpha_1,\ldots,\alpha_m\}$ is a minimal collection of $k$ cochains, and let $\sigma\in C^j(\X)$ be an arbitrary cell where $0\leq j<k.$
Unwinding the definitions shows that for $\gamma\in C^{*}(\X_\sigma)$
\begin{equation}\label{eq:when_lift_and_bdry_commute}
\delta(I^\sigma(\gamma))=I^\sigma(\delta_\sigma(\gamma)).
\end{equation}
Thus, from the minimality of the collection $\valpha,$ for any collection $\gamma_1,\ldots,\gamma_m\in C^{k-j-2}(\X_\sigma),$
\[\nr\bigcup_a\alpha_a\nr\leq\nr \bigcup_a(\alpha_a+\delta(I^\sigma(\gamma_a)))\nr=
\nr\bigcup_a(\alpha_a+I^\sigma(\delta_\sigma(\gamma_a)))\nr.\]
Applying Observation \ref{obs:EK2.7.1}, part 3, we see that, for any $\sigma,\gamma_1,\ldots,\gamma_m$ as above,
\[\nr \bigcup_a I_\sigma(\alpha_a)\nr_\sigma\leq \nr \bigcup_a (I_\sigma(\alpha_a) +\delta_\sigma(\gamma_a))\nr_\sigma,\]which implies the local minimality.

For the second item, write $\alpha'_a=\alpha_a\cap\beta.$ $\alpha'_a$ is contained in $\alpha_a.$
We will show that for a collection $\gamma_1,\ldots,\gamma_m\in C^k(\X),$
\begin{equation}\label{eq:for_restriction}\nr \bigcup_a \alpha'_a\nr\geq \nr \bigcup_a (\alpha'_a+\gamma_a)\nr \Rightarrow \nr \bigcup_a \alpha_a\nr\geq \nr \bigcup_a (\alpha_a+\gamma_a)\nr.\end{equation}
By applying \eqref{eq:for_restriction} to collections of coboundaries, we see that if $\{\alpha_1,\ldots,\alpha_m\}$ in minimal, then also  $\{\alpha'_1,\ldots,\alpha'_m\}$ is minimal.
We now prove \eqref{eq:for_restriction}. If $\nr \bigcup_a \alpha'_a\nr\geq \nr \bigcup_a (\alpha'_a+\gamma_a)\nr,$ then
\[\nr \left(\bigcup_a (\alpha'_a+\gamma_a)\right)\setminus \left(\bigcup_a \alpha'_a\right)\nr
\leq \nr \left(\bigcup_a \alpha'_a\right)\setminus \left(\bigcup_a (\alpha'_a+\gamma_a)\right)\nr.\]
The LHS equals
\begin{align*}&\nr \left(\left(\bigcup_a (\alpha_a+\gamma_a)\right)\setminus \left(\bigcup_a \alpha_a\right)\right)\cap \beta\nr+\nr(\bigcup_a\gamma_a)\cap\beta^c\nr\geq\\&\quad\quad
\nr \left(\left(\bigcup_a (\alpha_a+\gamma_a)\right)\setminus \left(\bigcup_a \alpha_a\right)\right)\cap \beta\nr+\nr \left(\left(\bigcup_a (\alpha_a+\gamma_a)\right)\setminus \left(\bigcup_a \alpha_a\right)\right)\cap \beta^c\nr=\\&\quad\quad\quad\quad\quad\quad\quad\quad\quad\quad\quad\quad\quad\quad\quad\quad
\quad\quad\quad\quad\quad\quad\quad\quad\quad\nr\left(\bigcup_a (\alpha_a+\gamma_a)\right)\setminus \left(\bigcup_a \alpha_a\right)\nr ,\end{align*}
where $\beta^c$ is the complement of $\beta$ in $X_k.$
The RHS is
\begin{align*}
\nr \left(\bigcup_a \alpha'_a\right)\setminus \left(\bigcup_a (\alpha'_a+\gamma_a)\right)\nr = \nr \left(\left(\bigcup_a \alpha_a\right)\setminus \left(\bigcup_a (\alpha_a+\gamma_a)\right)\right)\cap \beta\nr\leq \nr \left(\bigcup_a \alpha_a\right)\setminus \left(\bigcup_a (\alpha_a+\gamma_a)\right)\nr,
\end{align*}
thus
\[
\nr\left(\bigcup_a (\alpha_a+\gamma_a)\right)\setminus \left(\bigcup_a \alpha_a\right)\nr
\leq \nr \left(\bigcup_a \alpha_a\right)\setminus \left(\bigcup_a (\alpha_a+\gamma_a)\right)\Leftrightarrow
 \nr \bigcup_a (\alpha_a+\gamma_a)\nr \leq \nr \bigcup_a \alpha_a\nr
\]
and \eqref{eq:for_restriction} follows.

We turn to the third part. It is the collection version of the central proposition \cite[Proposition 2.5]{KKL} (see also \cite[Proposition 3.13]{EK}).
The proof is inductive and algorithmic. For a $k$-cochain $\alpha,$ write \[N(\alpha)=|X_d|\binom{d+1}{k+1}\nr\alpha\nr.\]
By definition $N(\alpha)$ is a non negative integer proportional to the weight  $\nr\alpha\nr$. Write $\valpha=\{\alpha_1,\ldots,\alpha_m\},$ $\alpha=\bigcup_a\alpha_a.$ We will induct on $N(\alpha).$ The claim clearly holds when $N(\alpha)=0,$ since then all $\alpha_a=0,$ so the collection is already locally minimal and each $\gamma_a$ can be taken to be $0.$
Suppose we have proven the claim for all collections $\valpha'=\{\alpha'_1,\ldots,\alpha'_m\}$ with $N(\bigcup_a\alpha'_a)<N(\alpha).$
If $\valpha$ is locally minimal, we can take all $\gamma_a$ to be $0,$ and the claim holds. Otherwise there is a $j-$cell $\sigma$ and cochains $\gamma'_1,\ldots,\gamma'_m\in C^{k-j-2}(\X_\sigma)$ with\[\nr\bigcup_a(I_\sigma(\alpha_a)+\delta_\sigma(\gamma'_a))\nr_\sigma<
\nr\bigcup_aI_\sigma(\alpha_a)\nr_\sigma.\]
By Observation \ref{obs:EK2.7.1}, part 3,
\[\nr\bigcup_a(\alpha_a+I^\sigma(\delta_\sigma(\gamma'_a))\nr<
\nr\bigcup_a \alpha_a\nr,\]
and the same holds when $\nr-\nr$ is replaced by $N(-).$ By \eqref{eq:when_lift_and_bdry_commute} the LHS equals $\nr\bigcup_a(\alpha_a+\delta(I^\sigma(\gamma'_a)))\nr.$
Thus,
\[N(\bigcup_a(\alpha_a+\delta(I^\sigma(\gamma'_a))))< N(\alpha).\] Since both sides of the equation are integers we can write
\begin{equation}\label{eq:dummy}N(\bigcup_a(\alpha_a+\delta(I^\sigma(\gamma'_a))))\leq N(\alpha)-1\Leftrightarrow
\nr\bigcup_a(\alpha_a+\delta(I^\sigma(\gamma'_a)))\nr\leq \nr\alpha\nr-\frac{Q}{|X_d|\binom{d+1}{k+1}}.\end{equation}
Applying the induction to the collection $\{\alpha_a''=\alpha_a+\delta(I^\sigma(\gamma'_a)\}_{a\in[m]}$ we can find a collection $\{\gamma_a''\}_{a\in[m]}$ with
\[\nr\bigcup_{a\in[m]}\gamma_a''\nr\leq Q\nr\bigcup_{a\in[m]}\alpha_a''\nr,\]
    the collection $\{\alpha'_1,\ldots,\alpha'_m\}=\{\alpha_1''+\delta\gamma_1'',\ldots,\alpha_m''+\delta\gamma_m''\}$
    is locally minimal and satisfies
    \[\nr\bigcup_{a\in[m]}\alpha'_a\nr\leq \nr\bigcup_{a\in[m]}\alpha_a''\nr.\]
Put $\gamma_a=\gamma_a''+I^\sigma(\gamma'_a),~a\in[m],$ then $\alpha'_a=\alpha_a+\delta\gamma_a,~a\in[m],$ and is locally minimal
\[\nr\bigcup \alpha'_a\nr\leq \nr\bigcup \alpha_a''\nr\leq\nr\bigcup \alpha_a\nr.\]
Finally,
\[\nr\bigcup_a\gamma_a\nr\leq \nr\bigcup_{a\in[m]}(\gamma_a''+I^\sigma(\gamma'_a))\nr\leq
\nr\bigcup_{a\in[m]}\gamma_a''\nr +\nr\bigcup_a I^\sigma(\gamma'_a)\nr\leq
Q\nr\bigcup_a\alpha_a''\nr+\frac{Q}{|X_d|\binom{d+1}{k+1}}\leq Q\nr\alpha\nr,\]
where the one before last passage used that all $\gamma'_a$ are contained in a the link, whose size is bounded by $Q.$ The last passage is \eqref{eq:dummy}.
\end{proof}

\begin{definition}\label{def:skeleton_exp}
A complex $\X$ is \emph{$\rho-$skeleton expander} if for any $A\subseteq X_0,$
\[\parallel E(A,A)\parallel\leq 4(\nr A\nr^2+\rho\nr A\nr),\]
where $E(A,A)$ is the set of edges between elements of $A.$
\end{definition}
\begin{ex}\label{ex:6.1}
\cite[Theorem 6.1]{EK} states that a $q$-thick Ramanujan complex\footnote{the thickness is the number of top cells which are incident to a given codimension-$1$ cell, minus $1$. The thickness is upper bounded by the locality.} of dimension $d$ is a $c'_dq^{-\frac{d-1}{2}}$-skeleton expander, for some universal constant $c'_d,$ which depends only on $d.$ \cite[Theorem 5.19]{EK} states that a $q-$thick spherical complex of dimension $d$ is a $c_d/\sqrt{q}-$skeleton expander, for some universal constant $c_d$ which depends only on $d.$
In both cases, for fixed $d,$ as $q$ tends to infinity, the skeleton-expansion constant $\rho$ tends to $0.$
\end{ex}
\begin{thm}\label{thm:loc_to_glob}
For any $d,Q\in\mathbb{N}$ and $\mu>0,$ there exist $\rho=\rho(d,\mu)$ and $\mus=\mus(d,Q,\mu)$ such that for any simplicial complex $\X$ which satisfies:
\begin{itemize}
\item The size of any proper link of $\X$ is at most $Q$,
\item All proper links of $\X$ have the collective cofilling property with constant $\mu,$
\item $\X$ and all of its proper links are $\rho$-skeleton expanders,
\end{itemize}
then $\X$ has the collective cofilling property for all $k\leq d-2,$ with
\[\mus_k(\X)\leq \mus.\]
\end{thm}
\begin{rmk}
\cite[Theorem 3.2]{EK} states that with the above conditions, but replacing collective cofilling with usual cofilling, $\X$ is a coboundary expander.
\end{rmk}
\begin{proof}[Proof of Theorem \ref{thm:strong_cofilling_Ramanujan}]
Fix $d.$ From Theorem \ref{thm:strong_cofilling_spherical} there exists $\mu=\mu_d$ such that all spherical buildings of rank $\leq d$ have the collective cofilling property with constant $\mu.$ Let $\rho=\rho(d,\mu)$ be the constant guaranteed by Theorem \ref{thm:loc_to_glob}. By Example \ref{ex:6.1}, for all large enough $q$ any $q$-thick Ramanujan complex and all of its proper links (which are spherical complexes of dimension smaller than $d$) are $\rho$-skeleton expanders. The sizes of all proper links of these complexes, and hence also all localities, are bounded by a universal constant $Q_{d,q}$ (see, e.g. \cite[Corollary]{EK}). Thus, they satisfy all conditions of Theorem \ref{thm:loc_to_glob}, and hence they have the collective cofilling property with constant $\mus(d,q)=\mus(d,\mu_d,Q_{d,q}),$ as claimed.
\end{proof}
Following \cite{KKL,EK}, the strategy for proving Theorem \ref{thm:loc_to_glob} will rely on a isoperimetric inequality.
\begin{thm}\label{thm:isoperimetric}
For any $d\in\mathbb{N}$ and $\mu>0,$ there exist $\rho=\rho(d,\mu),~\eta=\eta(d,\mu)$ and $\epsilon=\epsilon(d,\mu)$ such that for any simplicial complex $\X$ which satisfies:
\begin{itemize}
\item All proper links of $\X$ have the collective cofilling property with constant $\mu,$
\item $\X$ and all of its proper links are $\rho$-skeleton expanders,
\end{itemize}
then for any locally minimal collection $\alpha_1,\ldots,\alpha_m\in C^k(\X),~0\leq k\leq d-1,$ with $\parallel\bigcup\alpha_i\parallel\leq \eta,$
\[\parallel\bigcup\delta\alpha_i\parallel\geq\epsilon\parallel\bigcup\alpha_i\parallel.\]
\end{thm}
\begin{proof}[Proof of Theorem \ref{thm:loc_to_glob}]
The proof adjusts the corresponding argument from \cite{KKL} to collections.
Let $\eta,\epsilon,\rho$ be the constants provided by Theorem \ref{thm:isoperimetric}. Let $\mus=\max\{Q,\frac{1}{\eta}\}.$

Suppose $\X$ satisfies the requirements in the statement, i.e. its proper links are of size at most $Q$, its proper links have collective cofilling with constant $\mu$ and $\X$ and its links are $\rho$-skeleton expanders.
Let $\beta_1,\ldots,\beta_m\in C^{k+1}$ be a collection of coboundaries. Let $\alpha_1\ldots,\alpha_m$ be a minimal, hence also locally minimal, collection of $k-$cochains with $\delta\alpha_i=\beta_i$ and union $\alpha.$ We would like to show \begin{equation}\label{eq:ineq_for_l_t_g}
\nr\bigcup\alpha_a\nr\leq\mus\nr\bigcup\beta_a\nr.\end{equation}
If $\nr\bigcup\beta_a\nr\geq\eta,$ the inequality clearly holds, since $\nr\alpha\nr\leq 1.$ Otherwise, applying Lemma \ref{lem:restirction_of_minimal}, part 3, to $\{\beta_1,\ldots,\beta_m\},$ we can find $k-$cochains $\gamma_1,\ldots,\gamma_m,$ with \begin{equation}\label{eq:final_for_l_t_g}\nr\bigcup\gamma_a\nr\leq Q\nr\bigcup\beta_a\nr.
\end{equation} and $k+1$-cochains $\beta'_a=\beta_a+\delta\gamma_a,~a\in[m]$ which form a locally minimal collection with
\[\nr\bigcup_a\beta'_a\nr\leq\nr\bigcup_a\beta_a\nr<\eta.\]
Applying Theorem \ref{thm:isoperimetric} to the collection $\{\beta'_1,\ldots,\beta'_m\},$ we obtain
\[\nr\bigcup\delta(\beta'_a)\nr\geq\epsilon\nr\bigcup\beta'_a\nr.\] But $\delta\beta'_a=\delta\beta_a=0,$ since $\beta_a$ are coboundaries. Thus, all $\beta'_a$ must be $0,$ and $\delta\gamma_a=\beta_a.$ The collection $\gamma_1,\ldots,\gamma_a$ is a collection of preimages which satisfies \eqref{eq:ineq_for_l_t_g}, by \eqref{eq:final_for_l_t_g}. As needed.
\end{proof}
It remains to prove the isoperimetric inequality. For the proof we will recall the \emph{fat machinery} of \cite{EK} and adjust it to our needs.

We begin with a series of definitions and claims from \cite[Subsection 3.2]{EK}.
\begin{definition}\label{def:fat,ladder,deadend}
Fix $\xi\in(0,1),$ and a cochain $\alpha\in C^k(\X).$

For $-1\leq i\leq k$ we define recursively the set of \emph{fat $i$ cochains} $S^i_\xi(\alpha)$ by $S^k_\xi(\alpha)=\alpha,$ and
\[S^{i-1}_\xi(\alpha)=\{\sigma\in X_{i-1}\big|\nr I_\sigma(S^{i}_\xi(\alpha))\nr_\sigma\geq\xi^{2^{k-i}}\}.\]

A \emph{fat $i$ ladder} sitting on a fat $i-$cell $\sigma$ is a $k$ cell $\tau\in\alpha$ such that there exists a sequence \[\sigma=\sigma_i\subset\sigma_{i+1}\subset\ldots \subset\sigma_k=\tau,~\text{such that } \sigma_j\in S^j_\xi(\alpha).\] We denote by $L_\xi(\alpha,\sigma)$ the collection of fat ladder sitting on $\sigma$ and put $L_\xi(\alpha,i)=\bigcup_{\sigma\in S_\xi^{i}(\alpha)}L_\xi(\alpha,\sigma).$

A \emph{degenerate $k+1$-face} is a face $\tau\in X_{k+1}$ which contains two fat faces $\sigma,\sigma'\in S_\xi^i(\alpha) $ whose intersection is a $(i-1)-$face which is not fat. The collection of degenerate $k+1$-faces is denoted by $\Upsilon_\xi(\alpha).$

Given a collection $\valpha=\{\alpha_1,\ldots,\alpha_m\}$ with union $\alpha=\bigcup\alpha_a,$ we write $S^i(\valpha),L_\xi(\valpha,\sigma),L_\xi(\valpha,i),\Upsilon_\xi(\valpha)$ for $S^i(\alpha),L_\xi(\alpha,\sigma),L_\xi(\alpha,i),\Upsilon_\xi(\alpha)$ respectively
\end{definition}
Intuitively, fat $i-1-$faces are those faces which touch many fat $i$ faces, fat $i$ ladders are those $k$ which have a sequence of fat faces, ordered by containment, from a $i$ face, such that dimensions of consecutive elements in this sequence differ by $1.$
We shall need the following three results (Corollary 3.7, Lemma 3.10 and Proposition 3.12 of \cite{EK}). The first two are simple consequences of the definitions, the third one uses the skeleton expansion.
\begin{lemma}\label{lem:EK3.7}
Let $\X$ be a $d-$complex, $0<\xi<1,~\alpha\in C^k(\X)$ for $k\leq d.$
If $\nr \alpha\nr<\xi^{2^{k+1}}$ then $\emptyset$ is not a fat $-1$ face.
\end{lemma}
\begin{lemma}\label{lem:EK3.10}
Let $\X$ be a $d-$complex, $0<\xi<1,~\alpha\in C^k(\X)$ for $k\leq d.$ Consider $\sigma\in S^{i}(\alpha),~-1\leq i\leq k.$ Suppose $\tau\in X_{k+1}$ contains $\sigma_1,\sigma_2$ such that $\sigma_1\in\alpha$ and $\sigma_2\in L_\xi(\alpha,\sigma).$ Then either $\tau\in \Upsilon_\xi(\alpha)$ or $\sigma_1\in L_\xi(\alpha,\sigma\cap\sigma_1).$
\end{lemma}
\begin{prop}\label{prop:EK3.12}
Let $k<d\in\mathbb{N},~\xi,\rho\in (0,1)$ with $\rho<\xi^{2^{k+1}},$ and let $\X$ be a $d$-dimensional complex all of whose links, including $\X$ itself, are $\rho$-skeleton expanders. Then for any collection $\valpha=\{\alpha_1,\ldots,\alpha_m\}$ of $k-$cochains,
\[\nr\Upsilon_\xi(\valpha)\nr\leq (k+2)2^{k+4}\xi\nr\bigcup\alpha_a\nr.\]\footnote{In \cite{EK} this proposition was stated for a single cochain rather than a collection, but since both sides of the depend only on the union $\bigcup\alpha_i$ the statement here is equivalent.}
\end{prop}

The next proposition \ref{prop:fat_machinery} lower bounds $\nr L_\xi(\valpha,i-1)\nr$ in terms of $\nr L_\xi(\valpha,i)\nr,$ the union of boundaries $\nr\bigcup_a\delta(\alpha_a)\nr$ and the error $\nr\Upsilon_\xi(\valpha)\nr.$ It is the collection version of \cite[Proposition 3.11]{EK}, and the proof is also very similar.
\begin{prop}\label{prop:fat_machinery}
Fix $k<d,~0<\xi<1$ and $\mus>0.$
Then for any simplicial complex $\X$ of dimension $d$ such that all of whose proper links have the collective cofilling property with constant $\mus,$ and any locally minimal collection
$\vec{\alpha}=\{\alpha_1,\ldots,\alpha_m\}$ of $k-$cochains, the following inequality holds:
\[\nr L_\xi(\valpha,i)\nr\leq \mus\binom{k+2}{i+1}\left((k+2)\nr L_\xi(\valpha,i-1)\nr+ \nr\bigcup_{a\in[m]}\delta(\alpha_a)\nr+\nr \Upsilon_\xi(\valpha)\nr\right).\]
\end{prop}
\begin{proof}
Let $\alpha=\bigcup_a\alpha_a.$
Denote by $L_\xi(\alpha_a,i)=L_\xi(\valpha,i)\cap \alpha_a,$ then
\[\bigcup_{a\in[m]} L_\xi(\alpha_a,i)=L_\xi(\valpha,i).\]
Put $J=\bigcup_{\sigma\in S^i_\xi(\valpha)}\bigcup_{a\in[m]}I^\sigma\delta_\sigma I_\sigma L_\xi(\alpha_a,i).$
We first show that
\begin{equation}\label{eq:J_containment}
J\subseteq
\Gamma^{k+1}( L_\xi(\valpha,i-1))\cup\bigcup_a\delta(\alpha_a)\cup \Upsilon_\xi(\valpha).
\end{equation}
Indeed, take any $\tau\in J,$ then $\tau\in I^\sigma\delta_\sigma I_\sigma L_\xi(\alpha_a,i),$ for some $a\in[m]$ and $\sigma\in S^i_\xi(\valpha).$
There are three possibilities.
\begin{enumerate}
\item If all $k$-faces of $\tau$ which belong to $\alpha_a$ contain $\sigma$ and are included in $L_\xi(\alpha_a,i)$ then $\tau\in\partial\alpha_a.$
\item The second possibility is that all $k$-faces of $\tau$ which belong to $\alpha_a$ contain $\sigma,$ but at least one of them, $\sigma_1,$ is not in $L_\xi(\alpha_a,i).$ In this case, since $\sigma_1\in\alpha_a$ and from the definition of $L_\xi(\alpha_a,i),$ we see that $\sigma_1\in \alpha\setminus L_\xi(\valpha,i).$  But as $\tau\in I^\sigma\delta_\sigma I_\sigma L_\xi(\alpha_a,i),$ there must exist a $k-$face of $\tau$ $\sigma_2\in L_\xi(\alpha_a,i)\subseteq L_\xi(\valpha,i).$ Applying Lemma \ref{lem:EK3.10} we see that $\tau\in \Upsilon_\xi(\valpha).$
\item The remaining possibility is that there is a $k-$face $\sigma_1$ of $\tau$ which belongs to $\alpha_a,$ hence to $\alpha,$ but does not contain $\sigma.$ Again, as $\tau\in I^\sigma\delta_\sigma I_\sigma L_\xi(\alpha_a,i),$ there must exist a $k-$face of $\tau$ $\sigma_2\in L_\xi(\alpha_a,i)\subseteq L_\xi(\valpha,i).$ Applying Lemma \ref{lem:EK3.10} again we see that either $\tau\in\Upsilon_\xi(\valpha)$ or $\sigma_1\in L_\xi(\valpha, \sigma_1\cap\sigma)\subseteq L_\xi(\valpha,i-1).$ In this case $\tau\in \Gamma^{k+1}(L_\xi(\valpha,i-1)).$
\end{enumerate}
\eqref{eq:J_containment} follows.
Using \eqref{eq:J_containment} and \eqref{eq:EK2.3} we can write
\begin{equation}\label{eq:J_cont_weights}
\nr J\nr\leq  (k+2)\nr L_\xi(\valpha,i-1)\nr+\nr\bigcup_a\delta(\alpha_a)\nr+ \nr\Upsilon_\xi(\valpha)\nr.
\end{equation}
We now use the collective cofilling to lower bound $\nr J\nr.$
Since $\valpha$ is a locally minimal collection, its restriction to each link is a minimal collection. By \ref{lem:restirction_of_minimal} also the collection
\[\{I_\sigma(\alpha_a)\cap I_\sigma(L_\xi(\valpha,\sigma))\}_{a\in[m]}=
\{I_\sigma(L(\alpha_a,\sigma))\}_{a\in[m]}\]is minimal.
By assumption the links have the collective cofilling property with constant $\mus.$
Thus, for all $\sigma\in S^i_\xi\subseteq X_i$
\[\nr I_\sigma(L_\xi(\valpha,\sigma))\nr_\sigma=\nr \bigcup_{a\in[m]}I_\sigma(L_\xi(\alpha_a,\sigma))\nr_\sigma
\leq \mus\nr \bigcup_{a\in[m]}\delta_\sigma I_\sigma(L_\xi(\alpha_a,\sigma))\nr_\sigma.\]
Using \eqref{eq:EK2.12}, noting that all cells of $L_\xi(\valpha,\sigma)$ contain $\sigma$ we see that
\[I^\sigma(I_\sigma(L_\xi(\valpha,\sigma)))= L_\xi(\valpha,\sigma).\]
Summing over $\sigma\in S^i_\xi(\valpha),$ and using Observation \ref{obs:EK2.7.1}, part 2, we obtain
\[\nr L_\xi(\valpha,i)\nr \leq \sum_{\sigma\in S^i_\xi(\valpha)}\nr L_\xi(\valpha,\sigma)\nr=\sum_{\sigma\in S^i_\xi(\valpha)}\nr\bigcup_{a\in[m]} L_\xi(\alpha_a,\sigma)\nr\leq\mus \sum_{\sigma\in S^i_\xi(\valpha)}\nr I^\sigma(\bigcup_{a\in[m]}\delta_\sigma I_\sigma(L_\xi(\alpha_a,\sigma)))\nr\]
Since any $\sigma\in S^i_\xi(\valpha)$ is contained in at most $\binom{k+2}{i+1}$ elements of $J,$ the left hand side is bounded by $\mus\binom{k+2}{i+1}\nr J\nr.$
Combining with \eqref{eq:J_cont_weights}, we arrive to
\[\nr L_\xi(\valpha,i)\nr \leq \mus\binom{k+2}{i+1}\nr J\nr\leq
\mus\binom{k+2}{i+1}\left((k+2)\nr L_\xi(\valpha,i-1)\nr+\nr\bigcup_a\delta(\alpha_a)\nr+ \nr\Upsilon_\xi(\valpha)\nr\right),\]and the claim follows.
\end{proof}

\begin{proof}[Proof of Theorem \ref{thm:isoperimetric}]
For a natural number $k,$ a positive real $\mu$ and arbitrary real numbers $\epsilon,\xi$ define the collection of constants $c^k_i=c^k_i(\mu,\epsilon,\xi),~-1\leq i\leq k$ by $c^k_k=1$ and
\[c^k_{i-1}=\frac{c^k_i-\epsilon-(k+2)2^{k+4}\xi}{(k+2)\mu\binom{k+2}{i+1}}.\]

Note that for any positive $\mu$ one can find positive $\xi,\epsilon$ such that all these constants are positive. Indeed, fixing $\mu\neq 0$ the constants are continuous functions of $\epsilon,\xi$ which are positive at $\epsilon=\xi=0.$ Choose, for $\mu$ as in the statement of the theorem $\epsilon,\xi>0$ for which those constants are all positive. Put $\eta=\xi^{2^{k+1}},$ and take any $\rho<\xi^{2^{k+1}}.$
Then using Proposition \ref{prop:EK3.12}, if $\X$ and its proper links are $\rho$-skeleton expanders, then $\nr\Upsilon_\xi(\valpha)\nr\leq (k+2)2^{k+4}\xi\nr\bigcup_a\alpha_a\nr.$ If the isoperimetric inequality for the collection $\valpha$ fails with the given $\epsilon,$ then by Proposition \ref{prop:fat_machinery}, $c_i^k\nr\bigcup_a\alpha_a\nr$ is a lower bound for $\nr L_\xi(\valpha,i)\nr.$ From the choice of $\xi,\epsilon$ we have in particular that $c_{-1}^k>0.$ But this contradicts Lemma \ref{lem:EK3.7}. Thus, the isoperimetric inequality is proven.
\end{proof}

\section{Tensoring Ramanujan complexes and explicit quantum LDPC codes of distance $\Omega(\sqrt{n}\log^kn)$ for any $k$}\label{sec:tensoramanujan}
\subsection{Systoles in tensor products of Ramanujan complexes}
Theorem 5.11 in \cite{EKZ} proves that Ramanujan complexes of degree $d$ with injectivity radius at least $R$ and non trivial $k-$th homology have $k-$systoles at least $c_{d,k}R^{k}.$ The proof uses only the fact that the universal cover of the (underlying topological space of the) complex is a Bruhat-Tits building, hence generalizes to our setting as well, since products of Bruhat-Tits buildings are Bruhat-Tits buildings as well (and the universal cover of a product is the product of universal covers).
We therefore immediately obtain:
\begin{lemma}\label{lem:polylog_systoles}
Let $\X_1,\ldots, \X_l$ be Ramanujan complexes of injectivity radius at least $R.$ Let \[\Y=\X_l\otimes \X_{l-1}\otimes\cdots\otimes \X_1,\]
then for every $k \leq \text{dim}{\Y},$\[\text{Sys}_k(Y)=\Omega( R^{k}),\]
the constant inside $\Omega$ depends on $k$ and $\text{dim}(\Y)$.
\end{lemma}
We will now provide a sketch of an alternative proof for this fact.
\begin{proof}[Sketch of proof]
Let $\Y$ be a tensor product of Ramanujan complexes whose injectivity radius is at least $R$ and whose dimension is $d.$ Denote by $Y_i$ the set of $i-$cells of $\Y.$ Let $\widetilde{\Y}$ be its universal cover, a Bruhat-Tits building. Let $\widetilde{\Y}_{\leq j}$ be its $j-$skeleton. Write $\widetilde{Y},~\widetilde{Y}_{\leq j}$ for $j\leq d$ for the underlying topological spaces. We refer the reader to \cite{AB} for more details about Bruhat-Tits buildings.

$\widetilde{Y}$ is a union of apartments, each apartment is a periodic tessellation of $\R^{d}$ and its geometric realization is homeomorphic to $\R^{d}.$ The union is along $\widetilde{Y}_{\leq d-1}.$  Using the tessellation it is straightforward that the apartments can be given an affine structure in a way that the intersections inherit the affine structure.
Using the tessellation again, the apartments can be given a Euclidean structure, compatible with intersections, which is quasi isometric with respect to the natural discrete metric induced by the cell structure of $\widetilde{\Y}$, and in addition $k-$cells have volumes bounded from above and below by universal constants which depend on $k,~d.$ The fact they depend only on $k,~d$ follows from periodicity of the tessellation and the following observation: The cells touching a vertex $v$ are cones over cells of lower dimensions in the link of $v.$ The link of $v$ is a spherical building whose apartments are spherical Coxeter complexes. These were completely classified, and there are finitely many isomorphism types in any dimension. Thus there are finitely many possible isomorphism type for the neighborhood of a vertex in each apartment, and this implies the claimed universal bounds. The quasi isometricity, again with universal upper and lower bounds for the ratio of metrics, follows from similar reasoning.

We fix the affine and Euclidean structures.
Any two points $x,y\in\widetilde{Y}$ belong to at least one apartment, therefore also the straight line which connects them belongs to the same apartment, and any apartment which contains them contains this straight line. This allows us to define the notion of a \emph{cone}.
Let $A$ be a subset of $\widetilde{Y},$ and $x\in \widetilde{Y},$ then the $Cone(x,A)$ is the union of line segments which connect $x$ to each $y\in A.$

We will prove the lemma by analyzing the minimal non trivial $k-$cycle. However, in order for this analysis to pursue we will need to extend our collection of chains to a larger family which is more flexible to local deformations and intersections with subspaces. We will work in the family of \emph{modulo $2$ flat $k-$cycles} in $Y, \widetilde{Y}$ and in their skeletons. Flat cycles in $\R^{n},$ with coefficients in finite groups, were introduced and studied by Fleming in \cite{F}. A gentle introduction to the subject is the lecture notes \cite{W}. Flat $k-$cycles have a notion of \emph{support} which in nice cases, such as polyhedral chains, are precisely the points contained in the (geometric realization of the) chain; they also have a notion of boundary, which is closely related to the standard notion of boundary; and they have the notion of \emph{mass} which generalizes the $k-$Hausdorff measure and hence the usual $k-$volume of polyhedral chains. The existence of a boundary allows to define a homology theory using these chains, and this theory turns out to be equivalent to the more standard homology theories.
The case where the coefficient group is $\Z_2$ is especially simple. One can extend this construction to mod $2$ flat chains in Bruhat-Tits buildings, in Ramanujan complexes, and in their skeletons after choosing the affine and Euclidean structures.
All the properties we will use are well known in the case of flat chains in $\R^{n},$ see the mentioned references, and extend to our case. We will identify a mod $2$ flat chain and its support.

Let $\zeta\in H_k(\Y)$ be the homology class of a $k-$systole of $\Y.$
From the compactness of the underlying topological space $Y,$ and compactness (in flat metric topology) of the family of mod $2$ flat cycles with mass bounded by any constant, one can show the existence of a mass minimizer $Z_{min}$ which is itself a flat mod $2$ cycle in the homology class $\zeta.$

Let $Z\in Z_k(\Y)$ be a cycle of homology class $\zeta$ whose number of cells is minimal. Then clearly
\[M_k(Z)\geq M_k(Z_{min}),\]
where $M_k$ is the $k-$mass.
Since the $k-$mass, or equivalently the $k-$volume, of each $k-$cell is bounded by a universal constant $C=C(d,k)>0$ we have
\[C\text{Sys}_k(\Y)\geq M_k(Z)\geq M_k(Z_{min}).\]
Therefore the proof will follow if we could prove a $\Omega(R^{k})$ lower bound for $M_k(Z_{min}).$

Let $m$ be the maximal dimension of a cell whose interior intersects the support of $Z_{min}.$ Then $Z_{min}$ is also the mass minimizer in homology class $\zeta$ among mod $2$ flat chains on $Y_{\leq m},$ which is the union of manifolds of dimension $m$ along the $m-1-$skeleton. In the Euclidean setting almost all points in the support of a mod $2$ flat $k-$chain with finite non zero mass are \emph{smooth}. This means that there is an approximate notion of a $k-$dimensional tangent space. It also means that (when $x$ is a smooth point in the flat chain $C$)
\begin{equation}\label{eq:ratio_r_to_0}\lim_{r\to 0}\frac{M_k(\{y\in C: |x-y|<r\})}{VB_k(r)}=1,\end{equation}
where $VB_{k}(r)=\omega_k r^{k}$ is the volume of the Euclidean $k-$ball of radius $r,$ and $\omega_k$ is the volume of the $k-$unit ball in $\R^{k}.$
For a mod $2$ flat chain in the $Y_{\leq m},$ the same holds within the $m-$skeleton, for almost all $x$ \emph{which belongs to the interior of a $m-$cell}.
Fix such $x.$

By quasi isometricity of the graph metric and the Euclidean metric induced from the $\widetilde{Y},$ and the definition of the injectivity radius, there is a positive universal constant $c=c(d,k),$ such that any lift of the ball $B(x,cR)$ to the universal cover $\widetilde{Y}$ maps bijectively to $B(x,cR).$ 
We may assume $c<1.$
Write \[Z_r=Z_{min}\cap B(x,r).\]
We will prove a monotonicity result for $r\leq cR:$
\begin{equation}\label{eq:monoton}
M_k(Z_r)/VB_{k}(r)~\text{is monotonically non decreasing}.\end{equation}
Our proof adapts the idea of \cite[Section 5]{GL} or \cite[Section 8]{W} to our setting.
Since $r\leq cR$ we may assume that $Z_r$ is a chain on the Bruhat-Tits cover $\widetilde{Y}.$

Using \eqref{eq:ratio_r_to_0}, if \eqref{eq:monoton} holds, then \[M_k(Z_{min})\geq Vol_k(Z_{cR})\geq VB_{k}(cR)=\omega_k(cR)^{k},\]
proving the lemma (with the universal constant inside the $\Omega$ being $\frac{\omega_kc^{k}}{C}$).

It remains to prove \eqref{eq:monoton}. The intersection of $Z_{min}$ with $B(x,r)$ is a flat mod $2$ $k-1-$chain for almost all $r.$ Moreover, the derivative of $M_k(Z_r)$ with respect to $r$ exists for almost all $r,$ and whenever it exists, it holds that
\begin{equation}\label{eq:dr}
M_{k-1}(\partial Z_r)\leq \frac{d}{dr}M_k(Z_r),
\end{equation}
where $\partial Z_r$ is the intersection of $Z$ with the $r-$sphere.

In the Euclidean case, for any $k-1-$chain \[Y\subseteq S^{m-1}(r):=\{v\in\R^{m}~s.t.~|v|=r\},\]
\[M_{k}(Cone(0,Y))=\frac{r}{k}M_{k-1}(Y),\]
see for example \cite[Section 5]{GL} for an elegant derivation of this equality.

Any $y\in B(x,r)$ belongs to some apartment which contains $x.$
We can write $\partial Z_r=\bigcup_{\mathcal{A}}\partial Z_r^{\mathcal{A}},$ where the union is over apartments containing $x,$ and $Z_r^{\mathcal{A}}$ is the intersection of $Z_r$ with $\mathcal{A}.$ From the compactness of $Y$ it is evident that we can take finitely many apartments $\mathcal{A}_1,\ldots,\mathcal{A}_h$ such that
\[\partial Z_r=\bigcup_{i=1}^{h}\partial Z_r^{\mathcal{A}_i}.\]
If we write \[Z_r^{i}=\left(Z_r\cap \mathcal{A}_i\right)\setminus\left(\bigcup_{j<i}\mathcal{A}_j\right),\] then
\[M_{k-1}(\partial Z_r)=\sum_{i=1}^{h} M_{k-1}(\partial Z_r^{i}),\]
For each $i\leq\ h$
\[M_{k}(Cone(x,\partial Z_r^{i}))=\frac{r}{k}M_{k-1}(\partial Z_r^{i}),\]
since $\partial Z_r^{i}$ is contained in the sphere of radius $r$ in the Euclidean ball in $\mathcal{A}_i.$
Thus,
\begin{equation}\label{eq:r/k}M_{k}(Cone(x,\partial Z_r))\leq\frac{r}{k}M_{k-1}(\partial Z_r),\end{equation}
the reason for the inequality is that different cones may intersect, since the apartments are not disjoint, and even have intersections of positive mass.

Now, by minimality of $Z_{min}$ it must hold that
\begin{equation}\label{eq:cone_improve}M_{k}(Z_r)\leq M_{k}(Cone(x,\partial Z_r)),\end{equation}
since otherwise removing $Z_r$ from $Z_{min}$ and gluing $Cone(x,\partial Z_r)$ would give a $k-$cycle of smaller mass. This cycle is in the same homology class as $Z_{min}$, since
\[\partial(Z_r-Cone(x,\partial Z_r))=0,\] so $Z_r-Cone(x,\partial Z_r)$ is a $k-$cycle contained in $B(x,r).$ But the balls in Bruhat-Tits buildings have trivial homology, hence $Z_r-Cone(x,\partial Z_r)$ must be a $k-$boundary.

Combining \eqref{eq:dr} with \eqref{eq:r/k},~\eqref{eq:cone_improve} we get that for almost every $r$
\[M_k(Z_r)\leq\frac{r}{k}\frac{d}{dr}M_k(Z_r).\]
For the Euclidean ball\[VB_k(r)=\frac{r}{k}\frac{d}{dr}VB_k(r).\]
Putting together we get
\[0\leq\frac{d}{dr}\frac{M_k(Z_r)}{VB_{k}(r)},\]
and \eqref{eq:monoton} follows.
\end{proof}
\subsection{Properties of tensor products of Ramanujan complexes}
We now prove a more general version of Theorem \ref{thm:2}.
\begin{thm}\label{thm:complexes_with_linear_cosys_and_log_power_sys}
Let $l<d$ be natural number and consider Ramanujan complexes $\X_1,\ldots, \X_l$ of dimensions $d_1,\ldots, d_l\geq d.$ Assume that there exists $m$ and constants $a,c,C,Q$ such that for each $i,$ the number of vertices of $\X_i$ is between $cm$ and $Cm$, for some $m,$ its locality is at most $Q,$ and its injectivity radius is at least $a\log{m}.$ Write $n=m^{l}.$ Then the tensor product \[\Y=\X_l\otimes \X_{l-1}\otimes\cdots\otimes \X_1,\] has the following properties:
\begin{itemize}
\item $\Y$ has $\Theta(n)$ vertices and locality at most $lQ$. In addition each $|Y_i|$ is $\Theta(n).$ The constants inside $\Theta$ depend on $l,c,C,Q.$
\item $H_k(\Y),H^{k}(\Y)\neq 0,~k\leq l.$
\item $\text{CoSys}^{k}(\Y)=\Omega(n)$ for all $k\leq l.$ The constant inside $\Omega$ depends on $l,c,C,Q$ as well as the cofilling constant and cosystolic bounds of the complexes $\X_i.$
\item $\Y$ has the collective cofilling property (and in particular the cofilling property) with a constant which depends only on $l,c,C,Q$ the cosystolic bounds and collective cofilling constants of the complexes $\X_i.$
\item $\text{Sys}_k(\Y)=\Omega(log^{k}(n))$ for $k\leq l,$ the constant inside $\Omega$ depends on $k,~a$ and the total dimension $\sum d_i.$
\end{itemize}
\end{thm}

\begin{proof}
The first claim is straightforward; for the 'In addition' part we use that for Ramanujan complexes, or more generally pure complexes, of locality at most $Q$ and dimension at most $d,$ the number of $k-$simplices is upper and lower bounded by two constants which depend on $Q,d$ times the number of vertices.

By Theorem \ref{thm:LSV and their properties} the first homology and cohomology of each $\X_i,$ with $\F_2$ coefficients, are non trivial. By Theorem \ref{thm:Kunneth} this implies that for every $s$ the $k$th cohomology and homology groups of \[\mathbf{R}_s=\X_s\otimes \X_{s-1}\otimes\cdots\otimes \X_1,\] for $k\leq s$ is non zero.

By Theorem \ref{thm:LSV and their properties} again, $\mathbf{R}_1$ and each $\X_i$ have linear cosystoles and has bounded cofilling constants.
Since
\[\mathbf{R}_{s+1}=\X_{s+1}\otimes \mathbf{R}_s,\] by iterating Theorem \ref{thm:linear_cosys} we obtain that $\text{CoSys}^{k}(\Y)$ is linear in $|Y_k|$ for each $k\leq l,$ where the constant of proportionality is at least some universal constant which depends on $Q,$ on the cofilling constants and cosystolic bounds of each $\X_i,$ and on $c,C$.

From Theorem \ref{thm:strong_cofilling_Ramanujan} and Remark \ref{rmk:different weights}, each $\X_i$ also has the collective cofilling property with respect to the Hamming norm. Iterating the third part of Theorem \ref{thm:linear_cosys}, shows the collective cofilling property for the product.

The claim about the systoles is just Lemma \ref{lem:polylog_systoles}.
\end{proof}

\subsection{Generating LDPC quantum codes of distance $\Omega(\sqrt{n}\log^{k}{n})$}

Finally we arrive to the construction of the claimed explicit LDPC quantum codes of distance $\Omega(\sqrt{n\log^{k}n}).$
The explicitness of the code stems from the fact that it is obtained by taking a tensor power of the explicitly known Ramanujan complexes. Then applying the balancing procedure of \cite{EKZ} to the tensored complex. As the balancing procedure of \cite{EKZ} maintains explicitness the resulting code is an explicit LDPC quantum code with the claimed parameters.

%

\begin{proof}[Proof of Corollary \ref{cor:good_codes}]
Take $\X_1=\X_2=\ldots= \X_k=\X$ be Ramanujan complexes of dimension $d>k$ satisfying the properties of Theorem \ref{thm:LSV and their properties}. Define $\Y$ as in Theorem \ref{thm:complexes_with_linear_cosys_and_log_power_sys}, extract its $k-$th homology and cohomology pair, and apply to it the balancing procedure of Theorem \ref{thm:balancing}. The resulting code has the prescribed distance and dimension.
\end{proof}


\appendix
\section{Strong cofilling for building-like complexes}\label{app:spherical}
In this section we recall the definition of building-like complexes, following \cite{LMM}. We then prove that these complexes have the collective cofilling property. The proof is a minor adaptation of the proof of \cite{LMM} for the usual cofilling property, which in turn generalizes Gromov's ideas from \cite{Gromov}.
As in Section \ref{sec:strong_coff_LSV} we shall work solely with the weight function \eqref{eq:lob_norm}, and consider collective cofilling constants with respect to it. Here, unlike in Section \ref{sec:strong_coff_LSV}, the results will not be correct with the Hamming or normalized Hamming weight, since the complexes we consider do not have a bounded locality, see Remark \ref{rmk:different weights}.

Let $\X$ be a simplicial complex of dimension $n,$ $G$ a subgroup of $\Aut(\X),$ and $S$ a set on which $G$ acts. For $0\leq k\leq n-1,$ write \[\mathcal{F}_k=X_k\times S,\] it is endowed with the $G-$action given by
\[g(\sigma,s)=(g\sigma,gs).\]
Let
\[\mathcal{B}=\{B_{\sigma,s} :~ -1\leq k\leq n-1,~(\sigma,s)\in \mathcal{F}_k\}\]
be a collection of subcomplexes of $\X$ with the property that for any two simplices $\sigma\subseteq\sigma',$ and any $s\in S$
\[\sigma\in B_{\sigma,s}\subseteq B_{\sigma',s}.\]
\begin{definition}\label{def:building_like}[Definition 1.2 in \cite{LMM}]
A $4-$tuple $(\X,S,G,\mathcal{B})$ as above is a \emph{building like complex} if it satisfies the following:
\begin{enumerate}
\item $G$ acts transitively on $X_n.$
\item $gB_{\sigma,s}=B_{g\sigma,gs}.$
\item For all $-1\leq i\leq k\leq n-1,$ and all $(\sigma,s)\in\mathcal{F}_k,$ the reduced homology
$\widetilde{H}_i(B_{\sigma,s})=0.$\footnote{The \emph{reduced homology} $\widetilde{H}_\bullet$ of a based chain complex $\X=(X_0,\ldots,X_d)$ is the $i-$th homology group of the complex obtained from augmenting $\X$ by a $-1$-level $\mathbb{Z}$ and defining the augmentation map $\partial_0:C_0(\X)\to \Z$ by $\sum n_i\sigma_i\mapsto \sum n_i.$ It is easy to verify that for $i>0,~\widetilde{H}_i\simeq H_i,$ while $H_0\simeq\widetilde{H}_0\oplus\Z.$}
\end{enumerate}
For $0\leq k\leq n-1,$ let
\[a_k = a_k(\X, S, G, \mathcal{B}) = max\{|G\sigma' \cap (B_{\sigma,s})_{k + 1}| :~\sigma' \in X_{k+1},~(\sigma,s) \in\mathcal{F}_k\},\]
where $G\sigma'$ is the $G-$orbit of $\sigma',$ and as usual $(B_{\sigma,s})_{k + 1}$ is the collection of $(k+1)-$simplices of $B_{\sigma,s}.$
\end{definition}
\begin{thm}\label{thm:strong_cofill_building_like}
For a building like complex $(\X, S, G, \mathcal{B})$ as above, and $0\leq k\leq n-1,$
\[\mus_k\leq \binom{n+1}{k+2}a_k.\]
\end{thm}
The proof is very similar to the proof of \cite[Theorem 1.3]{LMM}, and it should not come as a surprise. The proofs in \cite{Gromov,LMM} construct, for a given coboundary, a $\delta-$preimage by using cones, and averaging over them. A light-weight preimage is obtained by finding a cone that is good on average for all simplices in the given coboundary. In a sense this is a local construction. The fact that one can find such cones make it natural to guess that one can do the same for a collection of coboundaries, and this is indeed the case. On Ramanujan complexes, on the other hand, the proof of existence of light-weight preimages is non constructive, and more importantly, it is not evident to be local, hence one needs to adapt it more in order to obtain the result for collections.
\begin{proof}
As mentioned above, a key ingredient is the existence of \emph{cones}, cf. \cite[Proposition 2.1]{LMM}.
\begin{prop}
There exists a collection of \emph{cones}
\[\mathcal{C}=\{c_{\sigma,s}\subseteq X_{k+1}:~(\sigma,s)\in\mathcal{F}_k,~-1\leq k\leq n-1\},\]
such that for all $(\sigma,s)\in\mathcal{F}_k$\[\partial c_{\sigma,s}=\sigma+\sum_{i=0}^k c_{(\sigma_\setminus\{i\}),s},\]
where $(\sigma_\setminus\{i\})$ is the $(k-1)-$simplex obtained from $\sigma$ by omitting its $i-$th vertex.
\end{prop}
The cones are used to construct \emph{contraction operations} $\iota_s:C^{k+1}(\X)\to C^{k}(\X),$ for $s\in S,$ i.e. operations which satisfy, for a cochain $\beta\in C^{k+1}(\X),$ and a $k-$simplex $\sigma$
\[(\iota_s\beta)(\sigma)=\beta(c_{\sigma,s}).\]
\cite[Claim 2.2]{LMM} shows that $\iota_s$ satisfies
\begin{equation}\label{eq:s_preimage}
\delta^{k}\circ\iota_s+\iota_s\circ\delta^{k+1}=id_{C^{k+1}(\X)}.
\end{equation}

Let $\beta_1,\ldots, \beta_m$ be a collection of $(k+1)-$coboundaries. Then \eqref{eq:s_preimage} yields
\[\delta^k\iota_s\beta_i=\beta_i.\]
The proof will follow if we could upper bound
\[\min_{s\in S}\|\bigcup_{i\in [m]}\iota_s\beta_i\|\leq \binom{n+1}{k+2}a_k\|\bigcup_{i\in[m]}\beta_i\|.\]

We will prove this inequality with average instead of minimum.
For this again we mimic \cite{LMM} and set
\[\theta_k=\max_{\sigma\in X_{k+1}}\left(\frac{1}{|S|w(\sigma)}\sum_{(\eta,s)\in\mathcal{F}_k:\sigma\in c_{\eta,s}} w(\eta)\right).\]
We will show
\begin{equation}\label{eq:ineq_for_spherical}
\frac{1}{|S|}\sum_s\|\bigcup_{i\in [m]}\iota_s\beta_i\|\leq\theta_k\|\bigcup_{i\in[m]}\beta_i\|.
\end{equation}
This equation, together with \cite[Claim 2.4]{LMM} which states
\[\theta_k\leq \binom{n+1}{k+2}a_k,\] finishes the proof.
It remains to establish inequality \eqref{eq:ineq_for_spherical}. The proof is analogous to the proof of \cite[Claim 2.3]{LMM}, but for collections of coboundaries.
\begin{align*}
\frac{1}{|S|}\sum_s&\|\bigcup_{i\in [m]}\iota_s\beta_i\|
\\&=\frac{1}{|S|}\sum_s \sum_{\bigcup_{i\in [m]}\{\sigma\in X_k:~ \iota_s\beta_i(\sigma)\neq 0\}} w(\sigma)\\&
=\frac{1}{|S|}\sum_s \sum_{\bigcup_{i\in [m]}\{\sigma\in X_k:~ \beta_i(c_{\sigma,s})\neq 0\}} w(\sigma)\\&
\leq\frac{1}{|S|}\sum_s \sum_{ \bigcup_{i\in [m]}\{\sigma\in X_k:~ \text{supp}\beta_i\cap\text{supp}(c_{\sigma,s})\neq \emptyset\}} w(\sigma)\\&\leq
\frac{1}{|S|}\sum_{\tau\in\bigcup_{i\in[m]}\text{supp}(\beta_i)}\sum_s
\sum_{\tau\in\text{supp}(c_{\sigma,s})}w(\sigma)\\&
=\frac{1}{|S|}\sum_{\tau\in\bigcup_{i\in[m]}\text{supp}(\beta_i)}
\sum_{\{(\sigma,s)\in\mathcal{F}_k:~\tau\in\text{supp}(c_{\sigma,s})\}}w(\sigma)
\\&\leq \theta_k\sum_{\tau\in\bigcup_{i\in[m]}\text{supp}(\beta_i)}w(\tau)
\\&=\theta_k\|\bigcup_{i\in[m]}\beta_i\|.
\end{align*}
The first inequality follows from the observation that if $\beta_i(c_{\sigma,s})\neq 0$ then $\text{supp}(\beta_i),~\text{supp}(c_{\sigma,s})$ intersect. In the second inequality we pass from summing over $w(\sigma)$ once for all $s$ such that $\bigcup\text{supp}(\beta_i)$ intersects $c_{\sigma,s}$, to summing, for all $\sigma,~s,$ the term $w(\sigma)|(\bigcup_i \text{supp}(\beta_i))\cap c_{\sigma,s}|.$ The last inequality follows from the definition of $\theta_k.$ The equalities are straight forward.
\end{proof}

\begin{proof}[Proof of Theorem \ref{thm:strong_cofilling_spherical}]
In \cite[Subsection 3.2]{LMM}, it is shown that spherical buildings are building like complexes. In their notations, for the spherical building $\Delta=\Delta(G;B,N),$ which is associated to the rank $n+1$ BN pair $(B,N),$ the group $G$ is $\langle B,N\rangle,$ the set $S$ is the collection of $n-$simplices of $\Delta,$ and for a $k-$simplex $\sigma$ and an $n-$simplex $\theta\in S,$ $B_{\sigma,\theta}$ is the intersection of the apartments which contain both $\sigma,\theta.$ The first two properties of spherical buildings are well known, the third one is proven in \cite[Claim 3.5]{LMM}. It is shown that \[a_k\leq \binom{n+1}{k+2}|W|,\]
where $W$ is the associated Weyl group, and therefore if $\omega_n$ is the size of the largest Weyl group of rank $n+1,$ then by Theorem \ref{thm:strong_cofill_building_like} we immediately get that
\[\mus_k\leq \mus(n,k)=\binom{n+1}{k+2}^2\omega_n.\] As needed.
\end{proof}

\end{document}